\documentclass[final]{siamltex}
\usepackage{amsfonts,amsmath,amstext,amsopn,amsbsy,amscd}
\usepackage{amsxtra,amssymb,latexsym,mathrsfs}
\usepackage{eufrak}

\usepackage{epsfig}
\usepackage{pstricks,multido,pst-node}
\usepackage{xcolor}

\title{Transport of Charged Particles: Entropy Production and  Maximum Dissipation Principle}

\author{Chia-Yu Hsieh  \thanks{Department of Mathematics, National Taiwan University, Taipei, Taiwan 10617, email: {\tt b92201049@gmail.com}},
\and YunKyong Hyon \thanks{Division of Mathematical Models, National Institute for Mathematical Sciences, Daejeon, Republic of Korea 305-811 email: {\tt hyon@nims.re.kr}},
\and Hijin Lee  \thanks{Department of Mathematical Sciences, Korea Advanced Institute of
Science and Technology, Daejeon, Republic of Korea 305-701 email: {\tt hijin@kaist.ac.kr}},
\and Tai-Chia Lin \thanks{Institute of Applied Mathematical Sciences, Center for Advanced Study in Theoretical Sciences (CASTS), National Taiwan University, Taipei, Taiwan 10617,  email:{\tt  tclin@math.ntu.edu.tw}},
\and Chun Liu      \thanks{Department of Mathematics, Pennsylvania State University, University Park,  PA 16802, USA,
                            email: {\tt liu@math.psu.edu}}
}

\date{\today}
\begin{document}

\maketitle

\newcommand{\phie}{\phi_\epsilon}
\newcommand{\psie}{\psi_\epsilon}
\newcommand{\Cnk}{c_n^k}
\newcommand{\Cnkp}{c_n^{k+1}}
\newcommand{\Cnkpm}{c_n^{k+1,m}}
\newcommand{\Cnkpmp}{c_n^{k+1,m+1}}
\newcommand{\Cpk}{c_p^k}
\newcommand{\Cpkp}{c_p^{k+1}}
\newcommand{\Cpkpm}{c_p^{k+1,m}}
\newcommand{\Cpkpmp}{c_p^{k+1,m+1}}
\newcommand{\phik}{\phi^k}
\newcommand{\phip}{\phi^{k+1}}
\newcommand{\phipm}{\phi^{k+1,m}}
\newcommand{\phib}{\phi^{k+1,\bar{m}}}
\newcommand{\phipmp}{\phi^{k+1,m+1}}
\newcommand{\dsp}{\displaystyle}

\newcommand{\BE}{\begin{equation}}
\newcommand{\BEN}{\begin{equation*}}
\newcommand{\EE}{\end{equation}}
\newcommand{\EEN}{\end{equation*}}

\newtheorem{thm}{Theorem}
\newtheorem{lem}[thm]{Lemma}
\newtheorem{rmk}[thm]{Remark}

\begin{abstract}
In order to describe the dynamics of crowded ions (charged particles), we use an energetic variational approach to derive a modified Poisson-Nernst-Planck (PNP) system which includes an extra dissipation due to the effective velocity differences between ion species. Such a system has more complicated nonlinearities than the original PNP system but with the same equilibrium states. Using Galerkin's method and Schauder's fixed-point theorem, we develop a local existence theorem of classical solutions for the modified PNP system. Different dynamics (but same equilibrium states) between the original and modified PNP systems can be represented by numerical simulations using finite element method techniques.
\end{abstract}

\section{Introduction}
The dynamics of ion transport is important for the study of biophysics as it is involved in almost all biological activities.
The transport of charged particles (ions), by nature, is a multiscale problem. The competition of thermal fluctuation, in terms of entropy,  and molecular (Coulomb)  interactions mainly give intriguing and significant behaviors of the systems. Choices of the variables, in terms of energetic functionals and entropy production (dissipation) functionals, demonstrate specific physical situations or applications in consideration. By employing an energetic variational approach (see Section~\ref{sec1-1}), we can derive the original Poisson-Nernst-Planck (PNP) system (see Section~\ref{sec1-2}) which describe dilute ionic liquids~\cite{R1,RLW,RLZ}.

The transport of ions in biological environments are usually in non-ideal situations.
Ion channels often have characteristic property of very high density distributions of ions
that are crowded into tiny spaces with huge electric and chemical fields and forces of excluded volume (cf.~\cite{Er1, Er11, E1}). To describe the dynamics of crowded ions, the energy functional and the dissipation functional should be modified. For the energy functional, we combined the energy functional of the original PNP system with Lennard-Jones type (LJ) potential (similar to those  used for molecular dynamic simulations) and derive new PNP-type systems which captured certain
properties of selectivity of ion channels (cf.~\cite{EHL10, HLLE, HEL, LE-cms13}).

The dynamical systems for transport of ions involve various types of entropy production.
The classical PNP equation involves the entropy production, the dissipation, in terms of sum of
damping due to individual ion species.
In this study, we take into consideration of the extra dissipation due to a drag force between different  species.
This extra dissipative effect,  due to the drag between ion
species,  is incorporated into the  derivation of  a modified PNP system.
The entropy production of modified PNP mainly contributes to the dynamics of the system,
while the equilibrium states, which are determined by the free energy, remain the same.
In other applications of physics, such consideration had been taken into
account in the study of ion heating in a plasma flow (cf.~\cite{DSZ_pp2008}).

The modified PNP system has more complicated nonlinearities than the original PNP system but with the same equilibrium states. Using Galerkin's method and Schauder's fixed-point theorem, we develop a local existence theorem of classical solutions for the modified PNP system. Furthermore, different dynamics (but same equilibrium states) between the original and modified PNP systems can be represented by numerical simulations using finite element method techniques.

The rest of this paper is organized as follows: In Section~\ref{sec2}, we derive the modified PNP system. The local existence of the modified PNP system is proved in Section~\ref{sec3}. In Section~\ref{sec4}, we provide numerical results of the modified PNP system and comparisons to those of the original PNP system.

\section{General Diffusion for Transport of Charged Particles}\label{sec2}
In this section, we firstly introduce the energetic variation framework for diffusions and then apply it to derive
 the original PNP system. Such a framework can be employed to the problem of transport of ions in non-ideal, non-diluted situations.
 We derive a modified PNP system that takes into account of additional dissipation due to the effect of velocity differences between ion species.

\subsection{Energetic Variational Approaches for Diffusion}\label{sec1-1}

For an isothermal closed system, the combination of the First Law and the Second Law of Thermodynamics
yields the following energy dissipation law:
\begin{align}\label{diss_law}
\frac{d}{dt} E^{\rm total} = -\triangle,
\end{align}
where $E^{\rm total}$ is the sum of kinetic energy and total
Helmholtz free energy, and $\triangle$ is the
entropy production (energy dissipation rate in this case).
The choice of total energy functional and dissipation
functional, together with the kinematic (transport) relation of the variables
employed in the system, determines all the physics and the assumptions for problem.

The energetic variational approach is the precise framework to obtain the force balance equations from the general
dissipation law \eqref{diss_law}. In particular,  the Least
Action Principle (LAP) will determine the Hamiltonian
part of the system and the Maximum Dissipation Principle (MDP) for the dissipative part. Formally,
LAP states the fact that force multiplies distance is equal to the work, i.e.,
\begin{align}
\delta E = {\rm force} \times \delta x,
\end{align}
where $x$ is the position, $\delta$ is the variation (derivative) in
general senses. This procedure will give the
Hamiltonian part of the system and the conservative forces \cite{AbMa78,Ar89}. On
the other hand, MDP, by Onsager and Rayleigh \cite{On31,On31a,Ra73}, yields
dissipative forces of the system:
\begin{align}\label{eqn:dissipation}
\delta \frac{1}{2} \triangle = {\rm force} \times \delta {\dot x}.
\end{align}
The factor $1/2$ in (\ref{eqn:dissipation}) is  consistent with the choice of quadratic form of the ``rates", which in turn describes the
linear response theory for long-time near equilibrium dynamics {\cite{HDL,Ku76}.
For instance, we consider the following inhomogeneous diffusion equation
\begin{equation}
f_t = \nabla\cdot (b(x) \nabla (a(x) f)), \label{eqn:diffusion}
\end{equation}
where $a(x)$ and $b(x)$ are given positive functions depending only on space with certain regularity properties (for the sake of
demonstration in this paper, we assume them to be smooth functions).

In fact, we can start with the following energy dissipation law with prescribed (Helmholtz) free energy
and entropy production functionals:
\begin{equation}
\frac{d}{dt} \int f \log (a(x) f) \, dx = - \int \frac{1}{a(x)b(x)} f |u|^2 \, dx,
\end{equation}
where $f$ is a probability distribution function.
$u$ is the (effective) velocity of the dynamics, that is, for the flow map $x(X, t)$, we have
$u(x(X, t), t) = x_t (X, t)$ where $X$ is the reference coordinate.
Both $a(x)$ and $b(x)$ are given functions.
It is clear that $a(x)$ contributes to the final equilibrium of the system, while $b(x)$, after renormalization,
states for the dissipation rate.

The transport kinematics of the distribution function $f$ is just the conservation of mass law:
\begin{equation}
f_t + \nabla\cdot (uf) = 0.
\end{equation}

From the energetic variational approach point of view \cite{EHL10,XSL14},
this energy dissipation law includes all the physics of the system.
Employing the LAP, one takes the variation of the free energy functional (the integral on the left
hand side) with respect to
the flow map $x(X, t)$. At the same time, by MDP, one takes variation of the dissipation functional
(the integral on the right hand side) with respect to the velocity.
The total force balance, the summation of the two variations gives,
\begin{equation}
\frac{1}{a(x)} \nabla (a(x) f) = - \frac{1}{a(x)b(x)} fu.
\end{equation}
Combining these with the kinematic conservation of law equation of $f$, we obtain the general
inhomogeneous diffusion equation (\ref{eqn:diffusion}).
From the above derivation, we can see that there are two independent ingredients in inhomogeneous diffusion.
While $a(x)$ enters through free energy, $b(x)$, is more associated with dissipation.

\subsection{Poisson-Nernst-Planck (PNP) System}\label{sec1-2}
Here we recall PNP system that has been widely used to describe the transport of ionic solutions \cite{R1,RLW,RLZ}. Consider positive and negative ions with charge concentrations, $c_p$, $c_n$, respectively.The dissipative energy law of ion dynamics including Brownian motion of charged ions is given as
\begin{eqnarray}
 &&\frac{d}{dt} \int \left\{ k_B T ( c_p \ln c_p + c_n \ln c_n ) +\frac{\varepsilon}{2} | \nabla \phi |^2 \right\} dx \nonumber\\[-0.7em]
 &&\label{eqn:DissipEnergyLaw}\\[-0.7em]
 &=& -\int k_B T\left( D_pc_p\left| \frac{\nabla c_p}{c_p}  +\frac{z_p q}{k_B T} \nabla\phi \right|^2
 +D_n c_n \left| \frac{\nabla c_n}{c_n}  +\frac{z_n q}{k_B T} \nabla\phi \right|^2 \right) dx\nonumber
\end{eqnarray}
where $k_B$ is the Boltzmann constant, $T$ is the absolute temperature,
$\varepsilon$ is the dielectric constant of the medium, $\phi$ is the electrostatic potential,
$D_p$, $D_n$ are the diffusion constants
and $z_p$, $z_n$ are the valences, for positive, negative ions, respectively.
Then the Nernst-Planck equations for ion dynamics from the dissipative energy law (\ref{eqn:DissipEnergyLaw})
can be derived by the energetic variational approach with the following evolution equations \cite{TK93}:
\begin{eqnarray}
 \frac{\partial c_p}{\partial t}
   =  \nabla \cdot \left( \frac{D_p}{k_B T} c_p \nabla \mu_p \right),\quad
 \frac{\partial c_n}{\partial t}
   =  \nabla \cdot \left( \frac{D_n}{k_B T} c_n \nabla \mu_n \right) \label{eqn:PNPsystem_0}
\end{eqnarray}
where $\mu_p$, $\mu_n$ are the chemical potentials obtained by the variational derivatives
of the total energy, which is the left hand side in (\ref{eqn:DissipEnergyLaw}) with respect to the charge densities.
Explicit forms of the chemical potentials are given as $\mu_p = k_B T (\ln c_p +1) +z_p q \phi$ and $\mu_n = k_B T (\ln c_n +1) +z_n q \phi$.

The full system of equations for the dynamic of ion transport is then given by
\begin{eqnarray}
 \frac{\partial c_p}{\partial t}
  & = & \nabla \cdot \left\{ D_p \left( \nabla c_p +\frac{z_p q}{k_B T} c_p \nabla \phi \right) \right\},  \label{eqn:PNPsystem_0.1} \\
 \frac{\partial c_n}{\partial t}
  & = & \nabla \cdot \left\{ D_n \left( \nabla c_n +\frac{z_n q}{k_B T} c_n \nabla \phi \right) \right\},  \label{eqn:PNPsystem_0.2} \\
 \nabla \cdot (\varepsilon \nabla \phi) & = & -z_p q c_p -z_n q c_n\,, \label{eqn:PNPsystem_0.3}
\end{eqnarray}
which is called the original Poisson-Nernst-Planck (PNP) system. Existence and uniqueness theorems for such a system~\cite{BHN1994, G-zm1985, S-na1985} and a generalized PNP system coupled with Navier-Stokes equation~\cite{R1,RLW,RLZ} were developed in order to study the dynamics of ion transport, respectively.

The original PNP system can also be viewed as a special form of general diffusion,
which takes into account of particle-particle interaction through Coulomb forces \cite{XSL14}.
To demonstrate this,  we start with the following system of equations with some vector fields $\vec{u}_n$, $\vec{u}_p$.
\begin{eqnarray}
\frac{\partial c_n}{\partial t} + \nabla \cdot \left( c_n \vec{u}_n \right) = 0,\quad
\frac{\partial c_p}{\partial t} + \nabla \cdot \left( c_p \vec{u}_p \right) = 0 \label{eqn:system1}
\end{eqnarray}
satisfying the following free energy:
\begin{eqnarray}
 \mathcal{A} & = &\hspace{-0.3em} \int\hspace{-0.3em} \left\{k_B T (c_n \ln c_n \hspace{-0.2em}+\hspace{-0.2em} c_p \ln c_p)
 + \frac12 \hspace{-0.2em}\int\hspace{-0.2em} G^{\varepsilon}(x \hspace{-0.2em}-\hspace{-0.2em}y) (c_n \hspace{-0.2em}-\hspace{-0.2em}c_p)(x) (c_n \hspace{-0.2em}-\hspace{-0.2em}c_p)(y)dy\right\} dx,\label{eqn:FreeEnergy}
\end{eqnarray}
which corresponds to the total energy in (\ref{eqn:DissipEnergyLaw}) in a special case,
and has the entropy production $\triangle$, i.e., dissipation,
\begin{eqnarray}\label{eqn:Dissipation}
 \triangle & = & \int \left( \frac{k_B T}{D_n} c_n | \vec{u}_n |^2 +\frac{k_B T}{D_p} c_p | \vec{u}_p |^2 \right) dx.
\end{eqnarray}
The entropy production explains that the system is in a linear response region originated by the free energy (\ref{eqn:FreeEnergy}).

Then the force balance law between conservative and dissipative forces implies that
\begin{eqnarray}
 c_p \nabla \frac{\delta \mathcal{A}}{\delta c_p} = -\frac{k_B T}{D_p} c_p \vec{u}_p
 = -\frac12 c_p \frac{\delta \triangle}{\delta \vec{u}_p},\label{eqn:ForceBalanceLaw_p}\\
 c_n \nabla \frac{\delta \mathcal{A}}{\delta c_n} = -\frac{k_B T}{D_n} c_n \vec{u}_n
 = -\frac12 c_n \frac{\delta \triangle}{\delta \vec{u}_n}\label{eqn:ForceBalanceLaw_n}
\end{eqnarray}
that is,
\begin{equation}
 \frac{D_n}{k_B T} c_n\nabla \frac{\delta \mathcal{A}}{\delta c_n} = - c_n \vec{u}_n,\quad
 \frac{D_p}{k_B T} c_p\nabla \frac{\delta \mathcal{A}}{\delta c_p} = - c_p \vec{u}_p.\label{eqn:ForceBalanceLaw_s}
\end{equation}
From these derivations and manipulations, it is clear that while the original PNP system resembles to those
diffusion-drift equations, in fact, the only ingredient is diffusion, although being of nonlocal features.
Such an observation can be important in designing numerical algorithms as well as analysis.

\subsection{Modified PNP system: Entropy Production}\label{sec2.2}
From the above discussions on original PNP systems,
one easily see that the PNP system possesses a linear response of entropy production,
which describe the physical nature of near equilibrium of the whole system.
While the free energy includes all the information and properties of equilibrium states,
the dissipation functional, i.e., the entropy production governs the dynamics of the system.
Understanding  statistical physics and nonlinear thermodynamics properties of systems describing interactions
 between  different ion species
are extremely important  in order to obtain a realistic dynamic of ion transport,
especially for those with crowded ion populations,
which is very common in situations like  biological ion channels (cf.~\cite{Er1, Er11, E1}). The earlier studies of
such  nonlinear interactions
 had mostly focused on the total energy (cf.~\cite{EHL10, HLLE, HEL, LE-cms13}).

In what follows, we will consider modifying the entropy production rather than the total free energy.
It is clear such modification would only change the dynamics of the system when approaching the
equilibrium states, which are the same as those for classical PNP systems.

To take into account of dissipations due to  interaction between different species,
we will add a drag term as $\frac{k_B T}{D_{n,p}} c_n c_p | \vec{u}_n -\vec{u}_p |^2$ that
is due to the relative velocity differences to (\ref{eqn:Dissipation}).
The resulting modified entropy production becomes
\begin{eqnarray}\label{Dissipation2}
 \triangle^{\ast} = \int \left( \frac{k_B T}{D_n} c_n | \vec{u}_n^{\ast} |^2
  +\frac{k_B T}{D_p} c_p | \vec{u}_p^{\ast} |^2 +\frac{k_B T}{D_{n,p}} c_n c_p | \vec{u}_n^{\ast} -\vec{u}_p^{\ast} |^2 \right) dx.
\end{eqnarray}
The third term in the right hand side of (\ref{Dissipation2}) is a higher order correction
in terms of both densities and mobility constants. Note that we can
also choose the mobility coefficients $D_{n,p}$ for the higher order correction as one of (a) $\frac{D_n +D_p}{2}$: arithmetic average (b) $\frac{2 D_n D_p}{D_n +D_p}$: harmonic average (c) $\sqrt{D_n D_p}$: geometric average.

The same derivation as those of  (\ref{eqn:ForceBalanceLaw_p}), (\ref{eqn:ForceBalanceLaw_n})
will yield the force balance relations as:
\begin{eqnarray}\label{eqn:ForceBalanceLaw_2_1}
 c_n \nabla \frac{\delta \mathcal{A}}{\delta c_n} = -\frac12 c_n \frac{\delta \triangle^{\ast}}{\delta \vec{u}_n^{\ast}}
  & = & - \left\{ \frac{k_B T}{D_n} c_n \vec{u}_n^{\ast} +\frac{k_B T}{D_{n,p}} c_n c_p ( \vec{u}_n^{\ast} -\vec{u}_p^{\ast} ) \right\},\\
 c_p \nabla \frac{\delta \mathcal{A}}{\delta c_p} = -\frac12 c_p \frac{\delta \triangle^{\ast}}{\delta \vec{u}_p^{\ast}}
  & = & - \left\{ \frac{k_B T}{D_p} c_p \vec{u}_p^{\ast} +\frac{k_B T}{D_{n,p}} c_n c_p ( \vec{u}_p^{\ast} -\vec{u}_n^{\ast} ) \right\}.
\end{eqnarray}
Noticing the coefficient in front of $\vec{u}_p^{\ast}$ in the first equation are exactly equal the coefficient in front of $\vec{u}_n^{\ast}$ in the
second equation. This is the exact manifestation of Onsager's reciprocal relations \cite{On31,On31a} in transport of different charged species.

Solving for ion fluxes in these equations, then we have that
\begin{eqnarray}\label{eqn:ForceBalanceLaw_3_1}
 c_n \vec{u}_n^{\ast} & = & - \frac{(D_{n,p} + D_p c_n)c_n \vec{u}_n + D_n c_n c_p \vec{u}_p}{D_{n,p} + D_n c_p + D_p c_n}, \\
 c_p \vec{u}_p^{\ast} & = & - \frac{(D_{n,p} + D_n c_p)c_p \vec{u}_p + D_p c_p c_n \vec{u}_n}{D_{n,p} + D_n c_p + D_p c_n}.
\end{eqnarray}
Again as for (\ref{eqn:system1}), we utilize the conservation of mass equations for both $c_p$ and $c_n$:
\begin{eqnarray*}
\frac{\partial c_n}{\partial t} + \nabla \cdot \left( c_n \vec{u}_n^* \right) = 0,\quad
\frac{\partial c_p}{\partial t} + \nabla \cdot \left( c_p \vec{u}_p^* \right) = 0\,,
\end{eqnarray*}
to get the resulting modified Nernst-Planck equations as follows:
\begin{eqnarray}
  & & \frac{\partial c_p}{\partial t} = -\nabla \cdot \left( c_p \vec{u}_p^{\ast} \right) \nonumber \\
  & = & -\nabla \cdot \left[ \frac{(D_{n,p} + D_n c_p)c_p \vec{u}_p + D_p c_p c_n \vec{u}_n}{D_{n,p} + D_n c_p + D_p c_n} \right] \label{eqn:MPNPsystem_1.2} \\
  & = & \nabla \cdot \left[ \frac{(D_{n,p} + D_n c_p)D_p( \nabla c_p + \frac{z_p q}{k_B T} c_p \nabla \phi ) + (D_p c_p) D_n ( \nabla c_n + \frac{z_n q}{k_B T} c_n \nabla \phi )}{D_{n,p} + D_n c_p + D_p c_n} \right]\nonumber\\
  && \frac{\partial c_n}{\partial t} = -\nabla \cdot \left( c_n \vec{u}_n^{\ast} \right) \nonumber \\
  & = & -\nabla \cdot \left[ \frac{(D_{n,p} + D_p c_n)c_n \vec{u}_n + D_n c_n c_p \vec{u}_p}{D_{n,p} + D_n c_p + D_p c_n} \right] \label{eqn:MPNPsystem_1.1}\\
  & = & \nabla \cdot \left[ \frac{(D_{n,p} + D_p c_n)D_n ( \nabla c_n + \frac{z_n q}{k_B T} c_n \nabla \phi ) + (D_n c_n) D_p( \nabla c_p + \frac{z_p q}{k_B T} c_p \nabla \phi )}{D_{n,p} + D_n c_p + D_p c_n} \right].\nonumber
\end{eqnarray}

Without lost of generalities, we choose $D = D_n = D_p$. Then $D_{n,p} = D$ and have the modified PNP system as:
\begin{eqnarray}
 &&\hspace{-0.7em}\frac{\partial c_n}{\partial t} =\hspace{-0.2em}  \nabla \hspace{-0.1em}\cdot\hspace{-0.1em} \left\{\hspace{-0.2em} \frac{D ( 1 +c_n )}{1 +c_n +c_p} \hspace{-0.2em}\left(\hspace{-0.2em} \nabla c_n \hspace{-0.2em}+ \hspace{-0.2em} \frac{z_n q}{k_B T} c_n \nabla \phi \right) \hspace{-0.2em}+\hspace{-0.2em}\frac{D c_n}{1 +c_n +c_p}\hspace{-0.2em} \left(\hspace{-0.2em} \nabla c_p \hspace{-0.2em}+\hspace{-0.2em} \frac{z_p q}{k_B T} c_p \nabla \phi \hspace{-0.2em}\right)\hspace{-0.3em} \right\} \label{eqn:MPNPsystem_2.1} \\
 &&\hspace{-0.7em}\frac{\partial c_p}{\partial t} =\hspace{-0.2em} \nabla \hspace{-0.1em}\cdot\hspace{-0.1em} \left\{ \hspace{-0.1em}\frac{D ( 1 + c_p )}{1 +c_n +c_p} \hspace{-0.2em}\left(\hspace{-0.2em} \nabla c_p \hspace{-0.2em}+\hspace{-0.2em} \frac{z_p q}{k_B T} c_p \nabla \phi \right) \hspace{-0.2em}+\hspace{-0.2em}\frac{D c_p}{1 +c_n +c_p} \hspace{-0.2em}\left(\hspace{-0.2em} \nabla c_n \hspace{-0.2em}+\hspace{-0.2em} \frac{z_n q}{k_B T} c_n \nabla \phi \hspace{-0.2em}\right)\hspace{-0.3em} \right\} \label{eqn:MPNPsystem_2.2} \\
 &&\hspace{-0.7em}\nabla \cdot (\varepsilon \nabla \phi) = -z_n q c_n -z_p q c_p. \label{eqn:MPNPsystem_2.3}
\end{eqnarray}
The entropy production of the modified PNP system (\ref{eqn:MPNPsystem_2.1})--(\ref{eqn:MPNPsystem_2.3}) is
\begin{eqnarray}\label{eqn:Dissipation_4}
&& \triangle^{\ast}  = \int \left( \frac{k_B T}{D_n} c_n | \vec{u}_n^{\ast} |^2
  +\frac{k_B T}{D_p} c_p | \vec{u}_p^{\ast} |^2 +\frac{k_B T}{D_{n,p}} c_n c_p | \vec{u}_n^{\ast} -\vec{u}_p^{\ast} |^2 \right) dx \nonumber \\
& = & Dk_B T \hspace{-0.2em}\int \left\{\hspace{-0.2em}c_n\hspace{-0.2em} \left| \frac{1 +c_n}{1 +c_n +c_p} \hspace{-0.2em}\left(\hspace{-0.2em}\frac{\nabla  c_n}{c_n}  + \frac{z_n q}{k_B T }\nabla \phi\hspace{-0.2em}\right)\hspace{-0.2em} +\hspace{-0.2em}\frac{c_p}{1 +c_n +c_p}\hspace{-0.2em} \left(\hspace{-0.2em} \frac{\nabla c_p}{c_p}\hspace{-0.2em} +\hspace{-0.2em} \frac{z_p q}{k_BT} \nabla\phi \hspace{-0.2em}\right)\hspace{-0.2em} \right|^2 \right.\nonumber \\
  & & + c_p \left| \frac{1 +c_p}{1 +c_n +c_p} \hspace{-0.2em}\left(\hspace{-0.2em} \frac{\nabla c_p}{c_p} \hspace{-0.2em}+ \hspace{-0.2em} \frac{z_p q}{k_BT} \nabla\phi \hspace{-0.2em}\right) + \frac{c_n}{1 +c_n +c_p} \hspace{-0.2em}\left(\hspace{-0.2em}\frac{\nabla c_n}{c_n} \hspace{-0.2em}+\hspace{-0.2em} \frac{z_n q}{k_BT} \nabla\phi \hspace{-0.2em}\right)\hspace{-0.2em} \right|^2 \\
  & & + \left.c_n c_p \left| \frac{1}{1 +c_n +c_p}\hspace{-0.2em} \left( \hspace{-0.2em}\frac{\nabla c_n}{c_n}\hspace{-0.2em} +\hspace{-0.2em} \frac{z_n q}{k_BT}\nabla\phi \hspace{-0.2em}\right)\hspace{-0.2em} - \frac{1}{1 +c_n +c_p}\hspace{-0.2em} \left(\hspace{-0.2em} \frac{\nabla c_p}{c_p} \hspace{-0.2em}+\hspace{-0.2em} \frac{z_pq}{k_BT}\nabla\phi \hspace{-0.2em}\right) \hspace{-0.2em}\right|^2\right\} dx, \nonumber
\end{eqnarray}
while  the original entropy production  of the classical PNP system takes the form as:
\begin{eqnarray}\label{eqn:Dissipation_3}
 \triangle & = & \int \left( \frac{k_B T}{D_n} c_n \left| \vec{u}_n \right|^2 +\frac{k_B T}{D_p} c_p \left| \vec{u}_p \right|^2 \right) dx \nonumber \\
 & = & Dk_B T\int  \left(c_n \left| \frac{\nabla c_n}{c_n} + \frac{z_n q}{k_B T}\nabla \phi \right|^2 + c_p \left| \frac{\nabla c_p}{c_p}
 +\frac{z_p q }{k_BT} \nabla\phi \right|^2 \right)dx.
\end{eqnarray}

The resulting modified PNP system (\ref{eqn:MPNPsystem_2.1})--(\ref{eqn:MPNPsystem_2.3}) involves much more complicated nonlinear coupling between unknown variables. Comparing with the original PNP system of equations (\ref{eqn:PNPsystem_0.1})--(\ref{eqn:PNPsystem_0.3}), it brings
extra difficulties in analysis. In the next section, as a first step in our systematical studies,
 we present the proof of the local existence theorem of classical solutions for the modified PNP system (\ref{eqn:MPNPsystem_2.1})--(\ref{eqn:MPNPsystem_2.3}).

\section{Local Existence of Solutions for the Modified PNP}\label{sec3}
The modified PNP system (\ref{eqn:MPNPsystem_2.1})--(\ref{eqn:MPNPsystem_2.3})
posses rather complicated saturable nonlinear terms in the forms as:
$$
\frac{1+{{c}_{n}}}{1+{{c}_{n}}+{{c}_{p}}}\,,\:\frac{{{c}_{n}}}{1+{{c}_{n}}+{{c}_{p}}}\,,\:
\frac{1+{{c}_{p}}}{1+{{c}_{n}}+{{c}_{p}}}\,,\:\frac{{{c}_{p}}}{1+{{c}_{n}}+{{c}_{p}}}
$$
as coefficients, which are found in coupling $\nabla {{c}_{n}}+\frac{{{z}_{n}}q}{{{k}_{B}}T}{{c}_{n}}\nabla \phi $ and $\nabla {{c}_{p}}+\frac{{{z}_{p}}q}{{{k}_{B}}T}{{c}_{p}}\nabla \phi $
The coefficients are different from the original PNP system of equations (\ref{eqn:PNPsystem_0.1})--(\ref{eqn:PNPsystem_0.3}).
Formally, if $1\ll {{c}_{n}}\ll {{c}_{p}}$, then $0<\frac{1+{{c}_{n}}}{1+{{c}_{n}}+{{c}_{p}}}\,,\:\frac{{{c}_{n}}}{1+{{c}_{n}}+{{c}_{p}}}\ll 1$ and equation (\ref{eqn:MPNPsystem_2.1}) becomes degenerate parabolic. Similarly, if $1\ll {{c}_{p}}\ll {{c}_{n}}$, then $0<\frac{1+{{c}_{p}}}{1+{{c}_{n}}+{{c}_{p}}}\,,\:\frac{{{c}_{p}}}{1+{{c}_{n}}+{{c}_{p}}}\ll 1$ and equation (\ref{eqn:MPNPsystem_2.2}) becomes degenerate parabolic.
Both $1\ll {{c}_{n}}\ll {{c}_{p}}$ and $1\ll {{c}_{p}}\ll {{c}_{n}}$ can be excluded if ${{c}_{n}}$ and ${{c}_{p}}$ are nonnegative and bounded for $x\in \Omega $ and $t\in \left( 0,T \right)$.
However, the maximum principle of (\ref{eqn:MPNPsystem_2.1})--(\ref{eqn:MPNPsystem_2.3}) has not yet been proved.
Thus it is nontrivial to assure that ${{c}_{n}},{{c}_{p}}\ge 0$ for $x\in \Omega ,t\in \left( 0,T \right)$ if the initial data  ${{c}_{n,0}},{{c}_{p,0}}\ge 0$ for $x\in \Omega$.
The fact motivates us to find nonnegative and bounded solution of (\ref{eqn:MPNPsystem_2.1})--(\ref{eqn:MPNPsystem_2.3}) in a finite time interval.

We now develop a local existence theorem for the modified PNP system of equations (\ref{eqn:MPNPsystem_2.1})--(\ref{eqn:MPNPsystem_2.3}) using Galerkin's method and Schauder's fixed-point theorem.
Because Schauder's fixed-point theorem cannot give the uniqueness of fixed point, it seems impossible to prove the uniqueness by our argument. For simplicity of derivation, we may set $D=k_B=T=\varepsilon=q=1$, $z_n = -1$, and $z_p = 1$ for equations (\ref{eqn:MPNPsystem_2.1})--(\ref{eqn:MPNPsystem_2.3}). Let $\Omega$ be a smooth and bounded domain in $\mathbb{R}^d$, $d\leq 3$. Then the modified PNP system can be written as
\begin{align}
\frac{\partial c_n}{\partial t} &= \nabla \cdot \left[\frac{1}{1+c_n+c_p} \Big((1+c_n)(\nabla c_n - c_n\nabla \phi) + c_n (\nabla c_p + c_p\nabla \phi)\Big)\right], \label{id1}\\
\frac{\partial c_p}{\partial t} &= \nabla \cdot \left[\frac{1}{1+c_n+c_p} \Big((1+c_p)(\nabla c_p + c_p\nabla \phi) + c_p (\nabla c_n - c_n\nabla \phi)\Big)\right], \label{id2}\\
\Delta \phi &= c_n - c_p \,, \quad\hbox{ for }\quad x\in\Omega\,, t>0\,,\label{id3}
\end{align}
with no-flux boundary conditions of $c_n$ and $c_p$, and Robin boundary condition of $\phi$ as follows:
\begin{align}
(\nabla c_n - c_n\nabla \phi)\cdot\nu &= 0, \label{id4} \\
(\nabla c_p + c_p\nabla \phi)\cdot\nu &= 0, \label{id5} \\
\phi + \alpha \frac{\partial \phi}{\partial \nu} &= \phi_0\,,\quad\hbox{ for }\quad x\in\partial\Omega\,, t>0\,,\label{id6}
\end{align}
where $\alpha$ is a nonnegative constant, $\nu$ is the unit outer normal vector to $\partial\Omega$,
and $\phi_0 = \phi_1 + \alpha \frac{\partial \phi_1}{\partial \nu}$
on $\partial\Omega$ for some $\phi_1 \in H^2(\Omega)$. For the initial data, we assume that
\begin{align}
0\le c_n(\cdot,0) &= c_{n,0}\in L^\infty(\Omega)\,, \label{id7} \\
0\le c_p(\cdot,0) &= c_{p,0}\in L^\infty(\Omega)\,,  \label{id8}
\end{align}
and $\phi(\cdot,0)$ is uniquely determined by (\ref{id3}) at $t=0$ with (\ref{id6}).

In order to find a local solution of (\ref{id1})--(\ref{id8}) in a finite time interval $(0,t_1)$, we consider the fixed point problem of the following map:
\BE\label{F-mp1}
F((\bar{c}_n,\bar{c}_p)) =(c_n,c_p) \quad\hbox{ for }\quad (\bar{c}_n,\bar{c}_p)\in X\,,
\EE
where
\begin{align*}
X = \{(f,g): f,g \in L^4((0,t_1);L^2(\Omega))\}\,,\quad t_1>0\,
\end{align*}
with the following specific norm
\begin{align*}
\|(f,g)\|_X = \|f\|_{L^4((0,t_1);L^2(\Omega))} + \|g\|_{L^4((0,t_1);L^2(\Omega))},
\end{align*}
and $(c_n,c_p)$ is the solution of
\begin{align}
\frac{\partial c_n}{\partial t} &= \nabla \cdot \left\{\frac{1}{1+\bar{c}_{n}^*+\bar{c}_p^*} \Big((1+\bar{c}_{n}^*)(\nabla c_n - c_n\nabla \bar{\phi}) + \bar{c}_{n}^* (\nabla c_p + c_p\nabla \bar{\phi})\Big)\right\}, \label{id9} \\
\frac{\partial c_p}{\partial t} &= \nabla \cdot \left\{\frac{1}{1+\bar{c}_{n}^*+\bar{c}_p^*} \Big((1+\bar{c}_p^*)(\nabla c_p + c_p\nabla \bar{\phi}) + \bar{c}_p^* (\nabla c_n - c_n\nabla \bar{\phi})\Big)\right\}\,, \label{id10}
\end{align}
with the initial data (\ref{id7})--(\ref{id8}) and boundary conditions
\begin{align}
(\nabla c_n - c_n\nabla \bar{\phi}) \cdot\nu &= 0, \label{id11} \\
(\nabla c_p + c_p\nabla \bar{\phi}) \cdot\nu &= 0\,. \label{id12}
\end{align}
The system of equations (\ref{id9}) and (\ref{id10}) is a linear system of parabolic equations of $c_n$ and $c_p$.

Let here
\begin{align}
\bar{c}_n^* &= \min\{\max\{\bar{c}_n,0\},5M_0\}, \label{bcn-up} \\
\bar{c}_p^* &= \min\{\max\{\bar{c}_p,0\},5M_0\}, \label{bcp-up} \\
M_0 &=\max\{\|c_{n,0}\|_{L^\infty(\Omega)},\|c_{p,0}\|_{L^\infty(\Omega)},1\}\,, \label{df-m0}
\end{align}
and let $\bar{\phi}$ be the solution of $\Delta \bar{\phi }={{\bar{c}}_{n}}-{{\bar{c}}_{p}}$ in $\Omega$ with the boundary condition (\ref{id6}).

Let
\begin{eqnarray*}
u &=& c_n+c_p\,,\hspace{1cm}  v= c_n-c_p\,, \\
\bar{u}&=&\bar{c}_n+\bar{c}_p\,, \hspace{1cm} \bar{v}=\bar{c}_n-\bar{c}_p\,, \\
\bar{u}^*&=&\bar{c}_n^*+\bar{c}_p^*\,, \hspace{1cm} \bar{v}^*=\bar{c}_n^*-\bar{c}_p^*\,.
\end{eqnarray*}
Then by adding and subtracting equations (\ref{id9}) and (\ref{id10}),
the system of equations for $u$, $v$ is given by
\begin{align}
u_t &= \nabla \cdot (\nabla u - v\nabla \bar{\phi})\,, \label{id13} \\
v_t &= \nabla \cdot \Big(\frac{1}{1+\bar{u}^*} (\nabla v - u\nabla \bar{\phi}) + \frac{\bar{v}^*}{1+\bar{u}^*} (\nabla u - v\nabla \bar{\phi})\Big)\,. \label{id14}
\end{align}
with the boundary and initial conditions
\begin{align}
& (\nabla u - v\nabla \bar{\phi})\cdot\nu = 0, \label{id15} \\
& (\nabla v - u\nabla \bar{\phi})\cdot\nu = 0\,, \label{id16} \\
& u(x,0) = u_0 = c_{n,0} + c_{p,0}, \label{id16-1} \\
& v(x,0) = v_0 = c_{n,0} - c_{p,0}\,. \label{id16-2}
\end{align}
Since $u$ and $v$ are linear combinations of $c_n$ and $c_p$,
one can easily recover the solution $(c_n,c_p)$ of (\ref{id9})--(\ref{id12}) with the initial data (\ref{id7})--(\ref{id8}) from $(u,v)$ of (\ref{id13})--(\ref{id16-2}). By (\ref{bcn-up})-(\ref{df-m0}), we also obtain
\BE\label{uv-bar}
0<\frac{1}{1+10{{M}_{0}}}\le \frac{1}{1+{{{\bar{u}}}^{*}}}\le 1,\quad \left| \frac{{{{\bar{v}}}^{*}}}{1+{{{\bar{u}}}^{*}}} \right|\le 1\,,
\EE
which are crucial for the study  of equation (\ref{id14}). Note that $0\le {{\bar{u}}^{*}}\le 10{{M}_{0}}$ and $\left| {{{\bar{v}}}^{*}} \right|\le {{\bar{u}}^{*}}$ because of $0\le \bar{c}_{n}^{*},\bar{c}_{p}^{*}\le 5{{M}_{0}}$.

The apriori estimate of the solution of (\ref{id13})-(\ref{id16-2}) is given as follows:\\
\begin{lem} \label{lem1}
Let $(u,v)$ be the solution of (\ref{id13})-(\ref{id16-2}). Then there exist positive constants $K_1, K_2$ and $\gamma$ depending only on $\alpha$, $M_0$, $d$, and  $\Omega$ such that
\begin{align}
&\frac{d}{dt} \int_\Omega (K_1 u^2 + v^2) dx + \gamma \int_\Omega (|\nabla u|^2 + |\nabla v|^2) dx \notag \\
&\quad\quad\leq K_2 \left(\int_\Omega (K_1 u^2 + v^2) dx\right)\left(1 + (\|\bar{v}\|_{L^2(\Omega)}^2 + \|\phi_1\|_{H^2(\Omega)}^2)^2\right). \label{id17}
\end{align}
Note that $\phi_1 \in H^2(\Omega)$ satisfies $\phi_1 + \alpha \frac{\partial \phi_1}{\partial \nu}=\phi_0$ on $\partial\Omega$, where $\phi_0$ and $\alpha$ come from the Robin boundary condition (\ref{id6}). Moreover, $\bar{v}= \bar{c}_n-\bar{c}_p=\Delta \bar{\phi }$ in $\Omega$.
\end{lem}\\
\begin{proof}
Multiply (\ref{id13}) by $u$ and integrate it over $\Omega$. Then using integration by parts and (\ref {id15}), we get
\begin{align}
\frac{1}{2} \frac{d}{dt} \int_\Omega u^2 dx = -\int_\Omega (|\nabla u|^2 - v\nabla\bar{\phi}\cdot\nabla u) dx, \label{id18}
\end{align}
In order to estimate the last term on the right hand side of (\ref{id18}), we need the interpolation inequality
\begin{align}
\|v\|_{L^3(\Omega)} \leq C \|v\|_{L^2(\Omega)}^{1/2} \|v\|_{H^1(\Omega)}^{1/2} \label{id19-1}
\end{align}
and Sobolev embedding theorem with the estimate for Poisson's
equation~\cite{GT1983}
\begin{align}
\|\nabla \bar{\phi}\|_{L^6(\Omega)}^2 \leq C \|\bar{\phi}\|_{H^2(\Omega)}^2 \leq C \left(\|\bar{v}\|_{L^2(\Omega)}^2 + \|\phi_1\|_{H^2(\Omega)}^2\right) \,.\label{id19-2}
\end{align}
For convenience, we use the same notation $C$ for a constant, which only depends on $\Omega$. Then using (\ref{id19-1}), (\ref{id19-2}), H\"{o}lder's and Young's inequalities, we have
\begin{align}
\Big|\int_\Omega (v\nabla\bar{\phi}\cdot\nabla u) dx\Big| &\leq \|\nabla u\|_{L^2(\Omega)} \|v\|_{L^3(\Omega)} \|\nabla \bar{\phi}\|_{L^6(\Omega)} \notag\\
&\leq C \|\nabla u\|_{L^2(\Omega)} \|v\|_{L^2(\Omega)}^{1/2} \|v\|_{H^1(\Omega)}^{1/2} \|\nabla \bar{\phi}\|_{L^6(\Omega)} \notag\\
&\leq \beta_1 \|\nabla u\|_{L^2(\Omega)}^2 + C(\beta_1) \|v\|_{L^2(\Omega)} \|v\|_{H^1(\Omega)} \|\nabla \bar{\phi}\|_{L^6(\Omega)}^2 \notag\\
&\leq \beta_1 \|\nabla u\|_{L^2(\Omega)}^2 + \beta_1 \|v\|_{H^1(\Omega)}^2 + C(\beta_1) \|v\|_{L^2(\Omega)}^2 \|\nabla \bar{\phi}\|_{L^6(\Omega)}^4 \notag\\
&= \beta_1 \|\nabla u\|_{L^2(\Omega)}^2 + \beta_1 \|\nabla v\|_{L^2(\Omega)}^2 + \|v\|_{L^2(\Omega)}^2 \left(\beta_1+C(\beta_1)\|\nabla \bar{\phi}\|_{L^6(\Omega)}^4\right) \notag\\
&\leq \beta_1 \|\nabla u\|_{L^2(\Omega)}^2 + \beta_1 \|\nabla v\|_{L^2(\Omega)}^2\notag\\[-0.8em]
& \label{id19}\\[-0.8em]
&\ \ \quad + \|v\|_{L^2(\Omega)}^2 \left(\beta_1+C(\beta_1) (\|\bar{v}\|_{L^2(\Omega)}^2 + \|\phi_1\|_{H^2(\Omega)}^2)^2\right)\notag
\end{align}
for ${{\beta }_{1}}>0$, where $C\left( {{\beta }_{1}} \right)>0$ is a constant  depending on ${{\beta }_{1}}$ and $\Omega$. Consequently,
\begin{align}
\frac{1}{2} \frac{d}{dt} \int_\Omega u^2 dx &\leq -(1-\beta_1) \|\nabla u\|_{L^2(\Omega)}^2 + \beta_1 \|\nabla v\|_{L^2(\Omega)}^2 \notag\\
&\ \ \quad+ \|v\|_{L^2(\Omega)}^2 \left(\beta_1+C(\beta_1) (\|\bar{v}\|_{L^2(\Omega)}^2 + \|\phi_1\|_{H^2(\Omega)}^2)^2\right). \label{n1}
\end{align}

As for (\ref{id18}), we multiply (\ref{id14}) by $v$ and integrate it over $\Omega$. Then we may use integration by parts and (\ref{id15})-(\ref{id16}) to get
\begin{align}
&\quad\ \frac{1}{2} \frac{d}{dt} \int_\Omega v^2 dx \label{id20}\\
&= -\int_\Omega \left\{\frac{1}{1+\bar{u}^*} (|\nabla v|^2 - u\nabla\bar{\phi}\cdot\nabla v) + \frac{\bar{v}^*}{1+\bar{u}^*} (\nabla u \cdot \nabla v - v\nabla \bar{\phi}\cdot\nabla v)\right\} dx. \notag
\end{align}
Notice that from~(\ref{uv-bar}),
\begin{align*}
0 < \frac{1}{1+10M_0} \leq \frac{1}{1+\bar{u}^*} \leq 1,\quad \left|\frac{\bar{v}^*}{1+\bar{u}^*}\right| \leq 1\,,
\end{align*}
which implies
\begin{align*}
\int_\Omega \frac{1}{1+\bar{u}^*} |\nabla v|^2 dx \geq \int_\Omega \frac{1}{1+10M_0} |\nabla v|^2 dx\,. 
\end{align*}
Besides, we may use Young's inequality to get
\begin{align*}
\Big|\int_\Omega \frac{\bar{v}^*}{1+\bar{u}^*} \nabla u \cdot \nabla v dx\Big| &\leq \int_\Omega |\nabla u \cdot \nabla v| dx \\
&\leq \beta_2 \|\nabla v\|_{L^2(\Omega)}^2 + C(\beta_2) \|\nabla u\|_{L^2(\Omega)}^2\, 
\end{align*}
for $\beta_2>0$, where $C\left( {{\beta }_{2}} \right)>0$ is a constant depending on ${{\beta }_{2}}$ and $\Omega$. On the other hand, as for (\ref{id19}), we have
\begin{align*}
\Big|\int_\Omega \frac{1}{1+\bar{u}^*} u\nabla\bar{\phi}\cdot\nabla v dx\Big| &\leq \int_\Omega |u\nabla\bar{\phi}\cdot\nabla v| dx \\
&\leq \beta_3 \|\nabla v\|_{L^2(\Omega)}^2 + \beta_3 \|\nabla u\|_{L^2(\Omega)}^2 \\
&\quad \ \ + \|u\|_{L^2(\Omega)}^2 \left(\beta_3+C(\beta_3) (\|\bar{v}\|_{L^2(\Omega)}^2 + \|\phi_1\|_{H^2(\Omega)}^2)^2\right),
\end{align*}
and
\begin{align*}
\Big|\int_\Omega \frac{\bar{v}^*}{1+\bar{u}^*} v\nabla\bar{\phi}\cdot\nabla v dx\Big| &\leq \int_\Omega |v\nabla\bar{\phi}\cdot\nabla v| dx \\
&\leq 2\beta_4 \|\nabla v\|_{L^2(\Omega)}^2 \\
& \hspace{0.5cm} + \|v\|_{L^2(\Omega)}^2 \left(\beta_4+C(\beta_4) (\|\bar{v}\|_{L^2(\Omega)}^2 + \|\phi_1\|_{H^2(\Omega)}^2)^2\right). 
\end{align*}
for $\beta_j>0, j=3,4$, where $C\left( {{\beta }_{j}} \right)>0$ is a constant depending on $\beta_j$ and $\Omega$. Hence
\begin{align}
\frac{1}{2} \frac{d}{dt} \int_\Omega v^2 dx &\leq -\left(\frac{1}{1+10M_0} - \beta_2 - \beta_3 - 2\beta_4\right) \|\nabla v\|_{L^2(\Omega)}^2 + \left(C(\beta_2)+\beta_3 \right) \|\nabla u\|_{L^2(\Omega)}^2 \notag\\
&\ \ \quad+ \|u\|_{L^2(\Omega)}^2 \left(\beta_3+C(\beta_3) (\|\bar{v}\|_{L^2(\Omega)}^2 + \|\phi_1\|_{H^2(\Omega)}^2)^2\right) \label{n2}\\
&\ \ \quad+ \|v\|_{L^2(\Omega)}^2 \left(\beta_4+C(\beta_4) (\|\bar{v}\|_{L^2(\Omega)}^2 + \|\phi_1\|_{H^2(\Omega)}^2)^2\right)\,, \notag
\end{align}
for $\beta_j>0, j=2,3,4$.

Combine (\ref{n1}) and (\ref{n2}) and then we get
\begin{align*}
{{K}_{1}}\left( \frac{1}{2}-{{\beta }_{1}} \right)-(C({{\beta }_{2}})+{{\beta }_{3}})\ge 0\,
\end{align*}
for sufficiently large $K_1$  and  sufficiently small $\beta_i$'s,
furthermore, by letting $\beta_1 = K_1^{-2}$, ${{\beta }_{2}}=\frac{1}{4\left( 1+10{{M}_{0}} \right)}$
and choosing $\beta_3, \beta_4$ small enough and $K_1$ large enough, we have that
\begin{align*}
\frac{1}{2\left( 1+10{{M}_{0}} \right)}-{{\beta }_{2}}-{{\beta }_{3}}-2{{\beta }_{4}}\ge {{K}_{1}}{{\beta }_{1}}\,.
\end{align*}
Then we obtain that
\begin{align}
&\quad\ \ \frac{1}{2}\frac{d}{dt} \int_\Omega \left(K_1 u^2 + v^2\right) dx + \frac{1}{2}\int_\Omega \left(K_1|\nabla u|^2 + \frac{1}{1+10M_0}|\nabla v|^2\right) dx \notag \\
&\leq \|u\|_{L^2(\Omega)}^2 \left(\beta_3+C(\beta_3) (\|\bar{v}\|_{L^2(\Omega)}^2 + \|\phi_1\|_{H^2(\Omega)}^2)^2\right) \label{id17-1}\\
&\quad\quad+ \|v\|_{L^2(\Omega)}^2 \left((K_1 \beta_1 + \beta_4) + (K_1 C(\beta_1) + C(\beta_4)) (\|\bar{v}\|_{L^2(\Omega)}^2 + \|\phi_1\|_{H^2(\Omega)}^2)^2\right). \notag
\end{align}
Note that choices of $K_1$ and $\beta_i$'s depend on $M_0$ and $\Omega$.

Therefore, by (\ref{id17-1}), we may get (\ref{id17}) and complete the proof of Lemma~\ref{lem1} by setting
\begin{align*}
&\gamma = \min\left\{K_1,\frac{1}{1+10M_0} \right\}, \\
&K_2 = 2\max\left\{\frac{\beta_3}{K_1},\frac{C(\beta_3)}{K_1},K_1 \beta_1 + \beta_4, K_1 C(\beta_1) + C(\beta_4)\right\}.
\end{align*}
\end{proof}

Now, we consider the weak solution of (\ref{id13})--(\ref{id16-2}), which satisfies
\begin{align}
&\int_{\Omega }{{{u}_{t}}wdx}+\int_{\Omega }{(\nabla u-v\nabla \bar{\phi })}\cdot \nabla wdx=0, \label{eq-wsol-u}\\
&\int_{\Omega }{{{v}_{t}}wdx}+\int_{\Omega }\left(\frac{1}{1+{{{\bar{u}}}_{*}}}(\nabla v-u\nabla \bar{\phi })+\frac{{{{\bar{v}}}_{*}}}{1+{{{\bar{u}}}_{*}}}(\nabla u-v\nabla \bar{\phi })\right)\cdot \nabla wdx=0 \,, \label{eq-wsol-v}
\end{align}
for $w\in {{H}^{1}}\left( \Omega  \right)$. There is no boundary integral terms in the weak forms (\ref{eq-wsol-u}) and (\ref{eq-wsol-v}) because of the natural boundary conditions (\ref{id15}) and (\ref{id16}) for (\ref{id13})--(\ref{id16-2}).
We now apply Galerkin's method (cf.~ Section 4--5 of Chapter III of~\cite{LSU1968}) to find the approximate solution of (\ref{eq-wsol-u})-(\ref{eq-wsol-v}) in the form of
\begin{align*}
u^m(x,t) &= \sum_{k=1}^m a_k^m(t) w_k(x) \\
v^m(x,t) &= \sum_{k=1}^m b_k^m(t) w_k(x),
\end{align*}
satisfying
\begin{align}
\int_\Omega u_t^m w_k dx &+ \int_\Omega (\nabla u^m - v^m\nabla \bar{\phi})\cdot\nabla w_k dx = 0, \label{eqn-um}\\
\int_\Omega v_t^m w_k dx &+ \int_\Omega \Big(\frac{1}{1+\bar{u}_*} (\nabla v^m - u^m \nabla \bar{\phi}) + \frac{\bar{v}_*}{1+\bar{u}_*} (\nabla u^m - v^m \nabla \bar{\phi})\Big) \cdot\nabla w_k dx = 0 \label{eqn-vm}
\end{align}
for $k = 1,2,...,m$ and $m\in \mathbb{N}$, where $\{w_k\}_{k=1}^\infty$ is an orthogonal basis of $H^1(\Omega)$ and an orthonormal basis of $L^2(\Omega)$. Hence the coefficients $a_k^m(t) = \int_\Omega u^m w_k dx$ and $b_k^m(t) = \int_\Omega v^m w_k dx$ can be determined by
\begin{align}
\frac{d}{dt}a_k^m(t) &+ \int_\Omega (\nabla u^m - v^m\nabla \bar{\phi})\cdot\nabla w_k dx = 0, \label{eqn-ak}\\
\frac{d}{dt}b_k^m(t) &+ \int_\Omega \Big(\frac{1}{1+\bar{u}_*} (\nabla v^m - u^m \nabla \bar{\phi}) + \frac{\bar{v}_*}{1+\bar{u}_*} (\nabla u^m - v^m \nabla \bar{\phi})\Big) \cdot\nabla w_k dx = 0 \label{eqn-bk}
\end{align}
for $t>0$ and
\begin{align*}
a_k^m(0) &= \int_\Omega u_0 w_k dx, \\
b_k^m(0) &= \int_\Omega v_0 w_k dx,
\end{align*}
for $k = 1, 2, \ldots, m$. (\ref{eqn-ak}) and (\ref{eqn-bk}) may form a system of ordinary differential equations so we may get the existence and uniqueness of $a_k^m$ and $b_k^m$ by the standard theorems of ordinary differential equations.

Multiply (\ref{eqn-um}), (\ref{eqn-vm}) by $a_k^m$, $b_k^m$, respectively,
and add them together for $k = 1,2,...,m$. Then we get
\begin{align*}
\frac{1}{2} \frac{d}{dt} \int_\Omega (u^m)^2 dx = -\int_\Omega (|\nabla u^m|^2 - v^m\nabla\bar{\phi}\cdot\nabla u^m) dx 
\end{align*}
and
\begin{align*}
&\quad\ \frac{1}{2} \frac{d}{dt} \int_\Omega (v^m)^2 dx \\ 
&= -\int_\Omega \left\{\frac{1}{1+\bar{u}^*} (|\nabla v^m|^2 - u^m\nabla\bar{\phi}\cdot\nabla v^m) + \frac{\bar{v}^*}{1+\bar{u}^*} (\nabla u^m \cdot \nabla v^m - v^m\nabla \bar{\phi}\cdot\nabla v^m)\right\} dx\,, 
\end{align*}
which have the same forms as (\ref{id18}) and (\ref{id20}), respectively. Then by the same argument of Lemma~\ref{lem1}, we have
\begin{align*}
&\frac{d}{dt} \int_\Omega (K_1 (u^m)^2 + (v^m)^2) dx + \gamma \int_\Omega (|\nabla u^m|^2 + |\nabla v^m|^2) dx \\
&\quad\quad\leq K_2 \left\{\int_\Omega (K_1 (u^m)^2 + (v^m)^2) dx\right\}\left\{1 + (\|\bar{v}\|_{L^2(\Omega)}^2 + \|\phi_1\|_{H^2(\Omega)}^2)^2\right\}\,,
\end{align*}
where $K_1, K_2$ and $\gamma$ are positive constants independent of $m$. This implies that by Gronwall's inequality,  $\left\{ {u^m} \right\}_{m=1}^{\infty }$ and $\left\{ {v^m} \right\}_{m=1}^{\infty }$ are uniformly bounded in $L^\infty((0,t_1);L^2(\Omega)) \cap L^2((0,t_1);H^1(\Omega))$. Therefore, we may find the solution of (\ref{eq-wsol-u})-(\ref{eq-wsol-v}) by setting $m\to \infty $ (up to a subsequence).

For the uniqueness of (\ref{id13})--(\ref{id16-2}), we may assume that $(u_1,v_1)$ and $(u_2,v_2)$ are solutions of (\ref{id13})--(\ref{id16-2}). Then $(u_1-u_2,v_1-v_2)$ is a solution of (\ref{id13})--(\ref{id16}) with zero initial data. By Lemma~\ref{lem1} and Gronwall's inequality, we have
\begin{align*}
\int_\Omega (K_1 (u_1 - u_2)^2 + (v_1-v_2)^2) dx \leq 0,
\end{align*}
which implies $u_1 \equiv u_2$, $v_1 \equiv v_2$. Hence (\ref{id13})--(\ref{id16-2}) have a unique solution. Equivalently, (\ref{id9})--(\ref{id12}) with initial data (\ref{id7})--(\ref{id8}) is uniquely solvable.

Therefore, $F$ is well-defined.

Now we claim the continuity of $F$ as follows: \\
\begin{lem} \label{lem2}
The map $F:X\to X$ defined at (\ref{F-mp1}) is continuous.
\end{lem}
\begin{proof}
Let $\left\{ ({{{\bar{c}}}_{n,k}},{{{\bar{c}}}_{p,k}}) \right\}_{k=1}^{\infty }\subset X$ and  $({{\bar{c}}_{n}},{{\bar{c}}_{p}})\in X$ such that $(\bar{c}_{n,k},\bar{c}_{p,k}) \rightarrow (\bar{c}_n,\bar{c}_p)$ in $X$ as $k\rightarrow \infty$. Let $(c_{n,k},c_{p,k})= F((\bar{c}_{n,k},\bar{c}_{p,k})) $ for $k\in \mathbb{N}$ and $ (c_n,c_p)= F((\bar{c}_n,\bar{c}_p)) $.

Claim that $(c_{n,k},c_{p,k})\to (c_n,c_p)$ as $k\rightarrow \infty$. As for (\ref{id13}) and (\ref{id14}), we may set
\begin{align*}
& u_k = c_{n,k} + c_{p,k}\,, & v_k = c_{n,k} -c_{p,k}\,, \\
& \bar{u}_k = \bar{c}_{n,k} + \bar{c}_{p,k}\,,  & \bar{v}_k = \bar{c}_{n,k} - \bar{c}_{p,k}\,, \\
& \bar{u}_k^* = \bar{c}_{n,k}^* +\bar{c}_{p,k}^*\,,  & \bar{v}_k^* = \bar{c}_{n,k}^* - \bar{c}_{p,k}^*\,, \\
\end{align*}
and
\begin{align*}
& u = c_n + c_p\,,  & v = c_n - c_p\,, \\
& \bar{u} = \bar{c}_n + \bar{c}_p\,,  & \bar{v} = \bar{c}_n - \bar{c}_p\,, \\
& \bar{u}^* = \bar{c}_n^* + \bar{c}_p^*\,, & \bar{v}^* = \bar{c}_n^* -\bar{c}_p^*\,.
\end{align*}
Then as for (\ref{id13})--(\ref{id16-2}), $\left(u_k, v_k\right)$ satisfies
\begin{align}
\frac{\partial u_k}{\partial t} &= \nabla \cdot (\nabla u_k - v_k \nabla \bar{\phi}_k), \label{id28} \\
\frac{\partial v_k}{\partial t} &= \nabla \cdot \Big(\frac{1}{1+\bar{u}_k^*} (\nabla v_k - u_k \nabla \bar{\phi}_k) + \frac{\bar{v}_k^*}{1+\bar{u}_k^*} (\nabla u_k - v_k \nabla \bar{\phi}_k)\Big) \label{id29}
\end{align}
with boundary conditions
\begin{align}
(\nabla u_k - v_k \nabla \bar{\phi}_k)\cdot\nu &= 0, \label{id32} \\
(\nabla v_k - u_k \nabla \bar{\phi}_k)\cdot\nu &= 0\,, \label{id33}
\end{align}
and $\left(u, v\right)$ do
\begin{align}
\frac{\partial u}{\partial t} &= \nabla \cdot (\nabla u - v\nabla \bar{\phi}), \label{id30} \\
\frac{\partial v}{\partial t} &= \nabla \cdot \Big(\frac{1}{1+\bar{u}^*} (\nabla v - u\nabla \bar{\phi}) + \frac{\bar{v}^*}{1+\bar{u}^*} (\nabla u - v\nabla \bar{\phi})\Big) \label{id31}
\end{align}
with boundary conditions
\begin{align}
(\nabla u - v\nabla \bar{\phi})\cdot\nu &= 0, \label{id34} \\
(\nabla v - u\nabla \bar{\phi})\cdot\nu &= 0, \label{id35}
\end{align}
and the initial data (\ref{id16-1})--(\ref{id16-2}), where $\bar{\phi}_k$ and $\bar{\phi}$ satisfy $\Delta {{\bar{\phi }}_{k}}={{\bar{c}}_{n,k}}-{{\bar{c}}_{p,k}}={{\bar{v}}_{k}}$ and $\Delta \bar{\phi }={{\bar{c}}_{n}}-{{\bar{c}}_{p}}=\bar{v}$ in $\Omega$, respectively, with the Robin boundary condition (\ref{id6}).

Let ${{\tilde{u}}_{k}}={{u}_{k}}-u$ and ${{\tilde{v}}_{k}}={{v}_{k}}-v$. Then by (\ref{id28})-(\ref{id31}),
we get the system of equations for ${{\tilde{u}}_{k}}$ and ${{\tilde{v}}_{k}}$ as follows:
\BE\label{eq-ti-uk}
\frac{\partial {{{\tilde{u}}}_{k}}}{\partial t}=\nabla \cdot \left( \nabla {{{\tilde{u}}}_{k}}+{{{\tilde{v}}}_{k}}\nabla {{{\bar{\phi }}}_{k}}+v\nabla \left( {{{\bar{\phi }}}_{k}}-\bar{\phi } \right) \right)\,,
\EE
and
\begin{align}
\frac{\partial\tilde{v}_k}{\partial t} = \nabla\cdot\bigg[ &\frac{1}{1+\bar{u}_{k}^{*}}\nabla {{{\tilde{v}}}_{k}}+\left( \frac{1}{1+\bar{u}_{k}^{*}}-\frac{1}{1+{{{\bar{u}}}^{*}}} \right)\nabla v \notag\\
&+\frac{\bar{v}_{k}^{*}}{1+\bar{u}_{k}^{*}}\nabla {{{\tilde{u}}}_{k}}+\left( \frac{\bar{v}_{k}^{*}}{1+\bar{u}_{k}^{*}}-\frac{{{{\bar{v}}}^{*}}}{1+{{{\bar{u}}}^{*}}} \right)\nabla u \label{eq-ti-vk}\\
&-\frac{1}{1+\bar{u}_{k}^{*}}\left( {{u}_{k}}\nabla {{{\bar{\phi }}}_{k}}-u\nabla \bar{\phi } \right)-\left( \frac{1}{1+\bar{u}_{k}^{*}}-\frac{1}{1+{{{\bar{u}}}^{*}}} \right)u\nabla \bar{\phi } \notag\\
&-\frac{\bar{v}_{k}^{*}}{1+\bar{u}_{k}^{*}}\left( {{v}_{k}}\nabla {{{\bar{\phi }}}_{k}}-v\nabla \bar{\phi }\right) -\left( \frac{\bar{v}_{k}^{*}}{1+\bar{u}_{k}^{*}}-\frac{{{{\bar{v}}}^{*}}}{1+{{{\bar{u}}}^{*}}} \right)v\nabla \bar{\phi}\ \ \bigg] \notag
\end{align}
Since (\ref{eq-ti-uk}) and (\ref{eq-ti-vk}) are similar to equations (\ref{id13}) and (\ref{id14}),
we can apply Lemma~\ref{lem1}, and then as for (\ref{n1}) in Lemma~\ref{lem1}, we have that
\begin{align}
\frac{1}{2} \frac{d}{dt} \int_\Omega \tilde{u}_k^2 dx
&\leq -(1-2\tilde{\beta}_1) \|\nabla \tilde{u}_k \|_{L^2(\Omega)}^2 + \tilde{\beta}_1 \|\tilde{v}_k\|_{H^1(\Omega)}^2 \notag\\
&\quad \ \ + C(\tilde{\beta}_1) \|\tilde{v}_k\|_{L^2(\Omega)}^2 (\|\bar{v}_k \|_{L^2(\Omega)}+\|\phi_1 \|_{H^2(\Omega)})^4  \label{id36-1}\\
&\quad \ \ + C(\tilde{\beta}_1) \|v\|_{L^2(\Omega)} \|v\|_{H^1(\Omega)} \|\bar{v}_k - \bar{v} \|_{L^2(\Omega)}^2 \notag
\end{align}
for $\tilde{\beta}_1>0$, where $C(\tilde{\beta}_1)>0$ is a constant depending on $\tilde{\beta}_1$ and $\Omega$.
Moreover, (\ref{eq-ti-vk}) gives
\BE\label{id37}
\frac{1}{2} \frac{d}{dt} \int_\Omega \tilde{v}_k^2 dx=I_1 + I_2\,
\EE
where
\begin{align*} {{I}_{1}}=&-\int_{\Omega }\left[ \frac{1}{1+\bar{u}_{k}^{*}}|\nabla {{{\tilde{v}}}_{k}}{{|}^{2}}+\frac{\bar{v}_{k}^{*}}{1+\bar{u}_{k}^{*}}\nabla {{{\tilde{u}}}_{k}}\cdot \nabla {{{\tilde{v}}}_{k}} \right]  \\
&+\int_{\Omega }{\left[ \frac{1}{1+\bar{u}_{k}^{*}}({{u}_{k}}\nabla {{{\bar{\phi }}}_{k}}-u\nabla \bar{\phi })\cdot \nabla {{{\tilde{v}}}_{k}}+\frac{\bar{v}_{k}^{*}}{1+\bar{u}_{k}^{*}}({{v}_{k}}\nabla {{{\bar{\phi }}}_{k}}-v\nabla \bar{\phi })\cdot \nabla {{{\tilde{v}}}_{k}} \right]}\,,  \\
\end{align*}
and
\begin{align*}
{{I}_{2}}=&-\int_{\Omega }{{}}\left[\Big(\frac{1}{1+\bar{u}_{k}^{*}}-\frac{1}{1+{{{\bar{u}}}^{*}}}\Big)\nabla v\cdot \nabla {{{\tilde{v}}}_{k}}+\Big(\frac{\bar{v}_{k}^{*}}{1+\bar{u}_{k}^{*}}-\frac{{{{\bar{v}}}^{*}}}{1+{{{\bar{u}}}^{*}}}\Big)\nabla u\cdot \nabla {{{\tilde{v}}}_{k}}\right] \\
& +\int_{\Omega }{{}}\left[\Big(\frac{1}{1+\bar{u}_{k}^{*}}-\frac{1}{1+{{{\bar{u}}}^{*}}}\Big)u\nabla \bar{\phi }\cdot \nabla{{{\tilde{v}}}_{k}}+\Big(\frac{\bar{v}_{k}^{*}}{1+\bar{u}_{k}^{*}}-\frac{{{{\bar{v}}}^{*}}}{1+{{{\bar{u}}}^{*}}}\Big)v\nabla \bar{\phi }\cdot \nabla {{{\tilde{v}}}_{k}}\right] \,.
\end{align*}

Since we may use the same method in Lemma \ref{lem1} to estimate $I_1$ like (\ref{id36-1}),
one can easily estimate for $I_1$. We omit the detail here.
For $I_2$, we may decompose the domain $\Omega$ into two parts as follows:
$$
\Omega\cap\left\{ \bigg|\frac{1}{1+\bar{u}_k^*} - \frac{1}{1+\bar{u}^*}\bigg| \leq \sigma \right\}\quad\hbox{and}\quad \Omega\cap\left\{ \bigg|\frac{1}{1+\bar{u}_k^*} - \frac{1}{1+\bar{u}^*}\bigg| > \sigma \right\}\quad\hbox{ for }\quad \sigma>0\,.
$$
Fix $\sigma>0$ arbitrarily. Then by Young's inequality, we have that
\begin{align*}
&\quad \ \ \bigg|\int_\Omega \Big(\frac{1}{1+\bar{u}_k^*} - \frac{1}{1+\bar{u}^*}\Big) \nabla v \cdot \nabla \tilde{v}_k dx\bigg| \notag \\
&\leq \tilde{\beta}_5 \int_\Omega |\nabla \tilde{v}_k|^2 dx + C(\tilde{\beta}_5) \int_\Omega \Big(\frac{1}{1+\bar{u}_k^*} - \frac{1}{1+\bar{u}^*}\Big)^2 |\nabla v|^2 dx \\
&\leq \tilde{\beta}_5 \int_\Omega |\nabla \tilde{v}_k|^2 dx + C(\tilde{\beta}_5)\left(\sigma^2 \int_\Omega |\nabla v|^2 dx + \int_{\Omega\cap\{ |\frac{1}{1+\bar{u}_k^*} - \frac{1}{1+\bar{u}^*}| > \sigma \}} |\nabla v|^2 dx\right)\,. \notag
\end{align*}

Similarly, we have that
\begin{align*}
&\bigg|\int_\Omega \Big(\frac{\bar{v}_k^*}{1+\bar{u}_k^*} - \frac{\bar{v}^*}{1+\bar{u}^*}\Big)\nabla u \cdot \nabla \tilde{v}_k dx\bigg| \\
&\leq \tilde{\beta}_6 \int_\Omega |\nabla \tilde{v}_k|^2 dx + C(\tilde{\beta}_6)\left(\sigma^2 \int_\Omega |\nabla u|^2 dx + \int_{\Omega\cap\{ |\frac{\bar{v}_k^*}{1+\bar{u}_k^*} - \frac{\bar{v}^*}{1+\bar{u}^*}| > \sigma \}} |\nabla u|^2 dx\right)\,,
\end{align*}
\begin{align*}
&\bigg|\int_\Omega \Big(\frac{1}{1+\bar{u}_k^*} - \frac{1}{1+\bar{u}^*}\Big) u\nabla\bar{\phi} \cdot\nabla \tilde{v}_k dx\bigg|   \\
&\leq \tilde{\beta}_7 \int_\Omega |\nabla \tilde{v}_k|^2 dx + C(\tilde{\beta}_7)\left(\sigma^2 \int_\Omega u^2 |\nabla \bar{\phi}|^2 dx + \int_{\Omega\cap\{ |\frac{1}{1+\bar{u}_k^*} - \frac{1}{1+\bar{u}^*}| > \sigma \}} u^2 |\nabla \bar{\phi}|^2 dx\right),
\end{align*}
and
\begin{align*}
&\bigg|\int_\Omega \Big(\frac{\bar{v}_k^*}{1+\bar{u}_k^*} - \frac{\bar{v}^*}{1+\bar{u}^*}\Big) v\nabla\bar{\phi} \cdot \nabla \tilde{v}_k dx\bigg|   \\
&\leq \tilde{\beta}_8 \int_\Omega |\nabla \tilde{v}_k|^2 dx + C(\tilde{\beta}_8)\left(\sigma^2 \int_\Omega v^2 |\nabla \bar{\phi}|^2 dx + \int_{\Omega\cap\{ |\frac{\bar{v}_k^*}{1+\bar{u}_k^*} - \frac{\bar{v}^*}{1+\bar{u}^*}| > \sigma \}} v^2 |\nabla \bar{\phi}|^2 dx\right)
\end{align*}
for $\tilde{\beta}_i>0$, where $C(\tilde{\beta}_i)>0$ is a constant depending on $\tilde{\beta}_i$ and $\Omega$, $i=5,6,7,8$. Hence (\ref{id37}) becomes
\begin{align}
&\quad\ \ \frac{1}{2} \frac{d}{dt} \int_\Omega \tilde{v}_k^2 dx \notag\\
&\leq -\left(\frac{1}{1+10M_0} -\tilde{\beta}_2 -2\tilde{\beta}_3 -3\tilde{\beta}_4 -\tilde{\beta}_5 -\tilde{\beta}_6 -\tilde{\beta}_7 -\tilde{\beta}_8\right)\|\nabla \tilde{v}_k \|_{L^2(\Omega)}^2 \notag\\
&\quad\ \ + \left(C(\tilde{\beta}_2) + \tilde{\beta}_3 \right)\|\nabla \tilde{u}_k \|_{L^2(\Omega)}^2 + \tilde{\beta}_4 \|\tilde{v}_k\|_{L^2(\Omega)}^2 \notag\\
&\quad\ \ + \left(C(\tilde{\beta}_3)\|\tilde{u}_k\|_{L^2(\Omega)}^2 + C(\tilde{\beta}_4)\|\tilde{v}_k\|_{L^2(\Omega)}^2 \right)\left(\|\bar{v}_k \|_{L^2(\Omega)}+\|\phi_0 \|_{L^2(\partial\Omega)}\right)^4 \notag\\
&\quad\ \ +\left(C(\tilde{\beta}_3)\|u\|_{L^2(\Omega)}\|u\|_{H^1(\Omega)}+C(\tilde{\beta}_4)\|v\|_{L^2(\Omega)}\|v\|_{H^1(\Omega)}\right)\|\bar{v}_k - \bar{v} \|_{L^2(\Omega)}^2 \label{id37-1}\\
&\quad\ \ + \sigma^2 \int_\Omega \left(C(\tilde{\beta}_6)|\nabla u|^2 + €C(\tilde{\beta}_5)|\nabla v|^2 + C(\tilde{\beta}_7)u^2|\nabla \bar{\phi}|^2 + C(\tilde{\beta}_8)v^2 |\nabla \bar{\phi}|^2 \right) dx \notag\\
&\quad \ \ +\int_{\Omega\cap\{ |\frac{1}{1+\bar{u}_k^*} - \frac{1}{1+\bar{u}^*}| > \sigma \}} \left(C(\tilde{\beta}_5)|\nabla v|^2 + C(\tilde{\beta}_7)u^2 |\nabla \bar{\phi}|^2\right) dx \notag \\
&\quad \ \ + \int_{\Omega\cap\{ |\frac{\bar{v}_k^*}{1+\bar{u}_k^*} - \frac{\bar{v}^*}{1+\bar{u}^*}| > \sigma \}} \left(C(\tilde{\beta}_6)|\nabla u|^2 + C(\tilde{\beta}_8)v^2 |\nabla \bar{\phi}|^2\right) dx. \notag
\end{align}
Combine (\ref{id36-1})--(\ref{id37-1}) and choose suitable $\tilde{K}$
large enough and $\tilde{\beta}_i$'s small enough such that
$
\tilde{K}(1-2\tilde{\beta}_1) - (C(\tilde{\beta}_2) + \tilde{\beta}_3) \geq 0
$
and
$$
\frac{1}{1+10M_0} -\tilde{\beta}_2 -2\tilde{\beta}_3 -3\tilde{\beta}_4 -\tilde{\beta}_5 -\tilde{\beta}_6 -\tilde{\beta}_7 -\tilde{\beta}_8 - \tilde{K}\tilde{\beta}_1 \geq 0.
$$

Set here $\tilde{\beta}_1 = \tilde{K}^{-2}$ and choose sufficiently large $\tilde{K}$ and sufficiently small $\tilde{\beta}_i$'s for $i = 2,3,...,8$ to get such $\tilde{K}$ and $\tilde{\beta}_i$'s. Then we have
\begin{align}
&\quad \ \ \frac{d}{dt} \int_\Omega \left(\tilde{K}\tilde{u}_k^2 + \tilde{v}_k^2\right) dx \notag \\
&\leq C\bigg[\left(1+(\|\bar{v}_k \|_{L^2(\Omega)}+\|\phi_0 \|_{L^2(\partial\Omega)}\right)^4)\int_\Omega \left(\tilde{K}\tilde{u}_k^2 + \tilde{v}_k^2\right) dx \notag \\
&\quad \ \ + \left(\|u\|_{L^2(\Omega)}\|u\|_{H^1(\Omega)}+\|v\|_{L^2(\Omega)}\|v\|_{H^1(\Omega)}\right)\|\bar{v}_k - \bar{v} \|_{L^2(\Omega)}^2 \label{id47}\\
&\quad \ \ + \sigma^2 \int_\Omega \left(|\nabla u|^2 + |\nabla v|^2 + u^2 |\nabla \bar{\phi}|^2 + v^2 |\nabla \bar{\phi}|^2\right) dx \notag \\
&\quad \ \ + \int_{\Omega\cap\{ |\frac{1}{1+\bar{u}_k^*} - \frac{1}{1+\bar{u}^*}| > \sigma \}} \left(|\nabla v|^2 + u^2 |\nabla \bar{\phi}|^2\right) dx \notag \\
&\quad \ \ + \int_{\Omega\cap\{ |\frac{\bar{v}_k^*}{1+\bar{u}_k^*} - \frac{\bar{v}^*}{1+\bar{u}^*}| > \sigma \}} \left(|\nabla u|^2 + v^2 |\nabla \bar{\phi}|^2\right) dx\bigg] \notag
\end{align}
for some positive constant $C$ depending only on $M_0$ and $\Omega$. By Gronwall's inequality, (\ref{id47}) implies
\begin{align}
&\quad \ \ \int_\Omega \left(\tilde{K}\tilde{u}_k^2 + \tilde{v}_k^2\right) dx \notag \\
&\leq C\exp\left\{C\int_0^t \big(1+(\|\bar{v}_k \|_{L^2(\Omega)}+\|\phi_0 \|_{L^2(\partial\Omega)})^4\big)dx\right\} \notag \\
&\quad \ \ \cdot \bigg[\int_0^t \left(\|u\|_{L^2(\Omega)}\|u\|_{H^1(\Omega)}+\|v\|_{L^2(\Omega)}\|v\|_{H^1(\Omega)}\right)\|\bar{v}_k - \bar{v} \|_{L^2(\Omega)}^2 ds \label{id48} \\
&\quad \ \ + \sigma^2 \int_{Q_t} \left(|\nabla u|^2 + |\nabla v|^2 + u^2 |\nabla \bar{\phi}|^2 + v^2 |\nabla \bar{\phi}|^2\right) dxds \notag \\
&\quad \ \ + \int_{Q_t \cap\{ |\frac{1}{1+\bar{u}_k^*} - \frac{1}{1+\bar{u}^*}| > \sigma \}} \left(|\nabla v|^2 + u^2 |\nabla \bar{\phi}|^2\right) dxds \notag \\
&\quad \ \ + \int_{Q_t \cap\{ |\frac{\bar{v}_k^*}{1+\bar{u}_k^*} - \frac{\bar{v}^*}{1+\bar{u}^*}| > \sigma \}} \left(|\nabla u|^2 + v^2 |\nabla \bar{\phi}|^2\right) dxds\bigg], \notag
\end{align}
where $Q_t := \Omega\times (0,t)$. Notice that
\begin{align*}
&\quad \ \ \int_0^{t_1}
\left(\|u\|_{L^2(\Omega)}\|u\|_{H^1(\Omega)}+\|v\|_{L^2(\Omega)}\|v\|_{H^1(\Omega)}\right)\|\bar{v}_k
- \bar{v} \|_{L^2(\Omega)}^2 ds \notag \\
&\leq \left(\|u\|_{L^\infty((0,t_1);L^2(\Omega))}\|u\|_{L^2((0,t_1);H^1(\Omega))}\right. \\
&\quad \ \ \left.+\|v\|_{L^\infty((0,t_1);L^2(\Omega))}\|v\|_{L^2((0,t_1);H^1(\Omega))}\right) \cdot \|\bar{v}_k - \bar{v}\|_{L^4((0,t_1);L^2(\Omega))}^2 \notag \\
&\rightarrow 0 \notag
\end{align*}
as $k\rightarrow\infty$. Using the following inequalities
\begin{align*}
\Big|\frac{1}{1+\bar{u}_k^*} - \frac{1}{1+\bar{u}^*}\Big| &= \frac{1}{(1+\bar{u}_k^*)(1+\bar{u}^*)} |\bar{u}_k^* - \bar{u}^*| \\
&\leq |\bar{u}_k^* - \bar{u}^*| \\
&\leq |\bar{c}_{n,k}^* - \bar{c}_n^*| + |\bar{c}_{p,k}^* - \bar{c}_p^*| \\
&\leq |\bar{c}_{n,k} - \bar{c}_n| + |\bar{c}_{p,k} - \bar{c}_p|
\end{align*}
and
\begin{align*}
\Big|\frac{\bar{v}_k^*}{1+\bar{u}_k^*} - \frac{\bar{v}^*}{1+\bar{u}^*}\Big| &= \frac{1}{(1+\bar{u}_k^*)(1+\bar{u}^*)} |(1+\bar{u}^*)\bar{v}_k^* - (1+\bar{u}_k^*)\bar{v}^*| \\
&\leq \frac{1}{(1+\bar{u}_k^*)}|\bar{v}_k^* - \bar{v}^*| + \frac{|\bar{v}^*|}{(1+\bar{u}_k^*)(1+\bar{u}^*)}|\bar{u}_k^* - \bar{u}^*| \\
&\leq |\bar{v}_k^* - \bar{v}^*| + |\bar{u}_k^* - \bar{u}^*| \\
&\leq 2(|\bar{c}_{n,k}^* - \bar{c}_n^*| + |\bar{c}_{p,k}^* - \bar{c}_p^*|) \\
&\leq 2(|\bar{c}_{n,k} - \bar{c}_n| + |\bar{c}_{p,k} - \bar{c}_p|)\,,
\end{align*}
we have
\begin{align*}
&\Big|Q_{t_1} \cap\Big\{\Big|\frac{1}{1+\bar{u}_k^*} - \frac{1}{1+\bar{u}^*}\Big| > \sigma\Big\}\Big| \rightarrow 0, \\
&\Big|Q_{t_1} \cap\Big\{\Big|\frac{\bar{v}_k^*}{1+\bar{u}_k^*} - \frac{\bar{v}^*}{1+\bar{u}^*}\Big| > \sigma\Big\}\Big| \rightarrow 0
\end{align*}
as $k \rightarrow \infty$.

Therefore, (\ref{id48}) implies that
\begin{align*}
& \limsup_{k\rightarrow\infty} \sup_{t\in(0,t_1)} \int_\Omega \left(\tilde{K}\tilde{u}_k^2 + \tilde{v}_k^2\right) dx \notag \\
&\leq C\sigma^2 \exp\left\{C\int_0^{t_1} \big(1+(\|\bar{v}\|_{L^2(\Omega)}+\|\phi_1 \|_{H^2(\Omega)})^4\big)dx\right\} \\
&\quad \ \ \cdot \int_{Q_{t_1}} \left(|\nabla u|^2 + |\nabla v|^2 + u^2|\nabla \bar{\phi}|^2 + v^2|\nabla \bar{\phi}|^2\right) dxdt\,. \notag
\end{align*}
In the derivation, we have used the assumption that $\bar{v}_k \rightarrow \bar{v}$ in ${L^4((0,t_1);L^2(\Omega))}$ as $k \rightarrow \infty$, and we complete the proof by letting $\sigma\to 0$.
\end{proof}

In order to use Schauder's fixed point theorem, we want to find a ball $B_R(0) =\{(f,g)\in X: \|(f,g)\|_X \leq R\}$ such that $B_R(0)$ is invariant under $F$ i.e., $G:=F({{B}_{R}}(0))\subseteq {{B}_{R}}(0)$ and the closure of $G$ is compact in $X$. The existence of such a ball can be proved as follows: \\
By Lemma~\ref{lem1} and Gronwall's inequality, we have
\begin{align}
&\sup_{0 \leq t \leq t_1} \int_\Omega (K_1 u^2 + v^2) dx \notag\\[-0.8em]
&\label{id26}\\[-0.8em]
&\leq \int_\Omega (K_1 u_0^2 + v_0^2) dx \cdot \exp\left\{K_2\int_0^{t_1} (1 + (\|\bar{v}\|_{L^2(\Omega)}^2 + \|\phi_1\|_{H^2(\Omega)}^2)^2) ds\right\}, \notag\\
&\int_{Q_{t_1}} (|\nabla u|^2 + |\nabla v|^2) dxdt\notag\\[-0.8em]
&\label{id27}\\[-0.8em]
&\leq \frac{1}{\gamma} \int_\Omega (K_1 u_0^2 + v_0^2) dx \cdot \left(1+K_2 t_1\exp\bigg\{2K_2\int_0^{t_1} (1 + (\|\bar{v}\|_{L^2(\Omega)}^2 + \|\phi_1\|_{H^2(\Omega)}^2)^2) ds\bigg\}\right),\notag
\end{align}
where $Q_{t_1} := \Omega\times(0,t_1)$. By (\ref{id26})--(\ref{id27}), we may estimate the norms of $u$ and $v$ in spaces $L^\infty((0,t_1);L^2(\Omega))$ and $L^2((0,t_1);H^1(\Omega))$. Moreover, (\ref{id26}) implies
\begin{align*}
\|(c_n,c_p)\|_X &\leq C\left(\int_0^{t_1} \Big(\int_\Omega (K_1 u^2 + v^2) dx\Big)^2 \right)^{1/4} \\
&\leq C_1 t_1^{1/4} \left(\|u_0\|_{L^2(\Omega)}+\|v_0\|_{L^2(\Omega)}\right) \exp\left\{C_2 (\|\bar{v}\|_{L^4((0,t_1);L^2(\Omega))}^4 + t_1(\|\phi_1\|_{H^2(\Omega)}^4 + 1))\right\} \\
&\leq C_1 t_1^{1/4} \left(\|u_0\|_{L^2(\Omega)}+\|v_0\|_{L^2(\Omega)}\right) \exp\left\{C_2 (\|(\bar{c}_n,\bar{c}_p)\|_X^4 + t_1(\|\phi_1\|_{H^2(\Omega)}^4 + 1))\right\}
\end{align*}
which implies that
${{\left\| F\left( {{{\bar{c}}}_{n}},{{{\bar{c}}}_{p}} \right) \right\|}_{X}}={{\left\| \left( {{c}_{n}},{{c}_{p}} \right) \right\|}_{X}}\le R$ if
\begin{align}
C_1 t_1^{1/4} \left(\|u_0\|_{L^2(\Omega)}+\|v_0\|_{L^2(\Omega)}\right) \exp\left\{C_2 (R^4 + t_1(\|\phi_1\|_{H^2(\Omega)}^4 + 1))\right\} \leq R\,, \label{R}
\end{align}
which can be fulfilled by fixing $R>0$ as a constant and letting $t_1>0$ sufficiently small such that
\begin{align*}
C_1 \left(\|u_0\|_{L^2(\Omega)}+\|v_0\|_{L^2(\Omega)}\right) t_1^{1/4} \exp\left\{ C_2 (\|\phi_1\|_{H^2(\Omega)}^4 + 1)) t_1 \right\} \leq R \exp\left\{-C_2 R^4\right\}.
\end{align*}

Therefore, we get the ball ${{B}_{R}}\left( 0 \right)$ as an invariant set of the map $F$.

Claim now that the image of the ball ${{B}_{R}}\left( 0 \right)$, $G := F(B_R(0))$ is precompact in $X$ i.e., the closure of $G$ is compact in $X$ as follows:
\begin{lem} \label{lem3}
The closure of the image $G:=F({{B}_{R}}(0))\subseteq {{B}_{R}}(0)$ of the ball $B_R(0) = \{(f,g):\|(f,g)\|_X \leq R)\}$ is compact in $X$, where $F$ is defined at (\ref{F-mp1}) and $R$ is defined in (\ref{R}) such that $B_R(0)$ is invariant under $F$.
\end{lem}
\begin{proof}
We may follow the proof of the standard PNP system (cf.~\cite{B1992} and \cite{BHN1994}). Equation (\ref{id9}) implies
\begin{align*}
\bigg|\left<\frac{\partial c_n}{\partial t},\eta\right>\bigg| &= \bigg|\int_\Omega \left[\frac{1}{1+\bar{c}_{n}^*+\bar{c}_p^*} \Big((1+\bar{c}_{n}^*)(\nabla c_n - c_n\nabla \bar{\phi}) + \bar{c}_{n}^* (\nabla c_p + c_p\nabla \bar{\phi})\Big)\right]\cdot\nabla\eta dx\bigg| \\
&\leq \int_\Omega \left(|\nabla c_n| + |c_n\nabla \bar{\phi}|+ |\nabla c_p| + |c_p\nabla \bar{\phi}|\right)|\nabla\eta|dx \\
&\leq \left(\|\nabla c_n\|_{L^2(\Omega)} + \|c_n\nabla \bar{\phi}\|_{L^2(\Omega)} + \|\nabla c_p\|_{L^2(\Omega)} + \|c_p\nabla \bar{\phi}\|_{L^2(\Omega)} \right)\|\nabla \eta\|_{L^2(\Omega)}.
\end{align*}
for any test function $\eta \in H^1(\Omega)$. By (\ref{id27}), $\|\nabla c_n\|_{L^2((0,t_1);L^2(\Omega))}$ and $\|\nabla c_p\|_{L^2((0,t_1);L^2(\Omega))}$ are uniformly bounded for $(c_n,c_p)\in G$. Moreover, by (\ref{id19-2}) and Holder's inequality, we may get $\|c_n\nabla \bar{\phi}\|_{L^2((0,t_1);L^2(\Omega))}$ and $\|c_p\nabla \bar{\phi}\|_{L^2((0,t_1);L^2(\Omega))}$ are uniformly bounded for $(c_n,c_p)\in G$. Consequently, $\|\frac{\partial c_n}{\partial t}\|_{L^2((0,t_1);H^{-1}(\Omega))}$ is uniformly bounded for $(c_n,c_p)\in G$.

Similarly, we have the uniform boundedness of $\|\frac{\partial c_p}{\partial t}\|_{L^2((0,t_1);H^{-1}(\Omega))}$.
Moreover, (\ref{id26}) and (\ref{id27}) give $c_n,c_p \in L^2((0,t_1);H^1(\Omega))$. Therefore, by Aubin-Lions lemma, $G$ is precompact in $L^2((0,t_1);L^2(\Omega))$ and also in $X =(L^4((0,t_1);L^2(\Omega)))^2$ because of the boundedness of $c_n,c_p$ in $L^\infty((0,t_1);L^2(\Omega))$.
\end{proof}

By Lemma \ref{lem2}, Lemma \ref{lem3}, and Schauder's fixed-point theorem, there exists a fixed point $(c_n,c_p)\in B_R(0)$ of $F$, which is a solution of
\begin{align}
\frac{\partial c_n}{\partial t} &= \nabla \cdot \left[\frac{1}{1+c_n^*+c_p^*} \Big((1+c_n^*)(\nabla c_n - c_n\nabla \phi) + c_n^* (\nabla c_p + c_p\nabla \phi)\Big)\right], \label{id55} \\
\frac{\partial c_p}{\partial t} &= \nabla \cdot \left[\frac{1}{1+c_n^*+c_p^*} \Big((1+c_p^*)(\nabla c_p + c_p\nabla \phi) + c_p^* (\nabla c_n - c_n\nabla \phi)\Big)\right] \label{id56}
\end{align}
with (\ref{id3})--(\ref{id8}), where
\begin{align*}
c_n^* &= \min\{c_{n+},5M_0\} = \min\{\max\{c_n,0\},5M_0\}, \\
c_p^* &= \min\{c_{p+},5M_0\} = \min\{\max\{c_p,0\},5M_0\}.
\end{align*}
We will now show $c_n^* = c_n$ and $c_p^* = c_p$ in a short time interval $(0,t_0)$ by the following lemma: \\
\begin{lem} \label{lem4}
The solution of (\ref{id55})--(\ref{id56}), (\ref{id3})--(\ref{id8})
satisfies $c_n,c_p \geq 0$ and $c_n+c_p \leq 5M_0$ for $0<t<t_0$ for
some $t_0 > 0$.
\end{lem}
\begin{proof}
Let $c_{n-} = \min\{c_n,0\}$. Then
\begin{align*}
\frac{1}{2}\frac{d}{dt}\int_\Omega c_{n-}^2 dx &= -\int_\Omega \frac{1}{1+c_p^*} (|\nabla c_{n-}|^2 - c_{n-} \nabla \phi \cdot \nabla c_{n-}) dx \notag \\
&\leq -\frac{1}{1+5M_0} \int_\Omega |\nabla c_{n-}|^2 dx + \|c_{n-}\|_{L^3(\Omega)} \|\nabla\phi\|_{L^6(\Omega)} \|\nabla c_{n-}\|_{L^2(\Omega)} \notag \\
&\leq -\frac{1}{1+5M_0} \int_\Omega |\nabla c_{n-}|^2 dx + C \|c_{n-}\|_{L^2(\Omega)}^{1/2} \|c_{n-}\|_{H^1(\Omega)}^{1/2} \|\nabla\phi\|_{L^6(\Omega)} \|\nabla c_{n-}\|_{L^2(\Omega)} \\
&\leq -\frac{1}{1+5M_0} \|\nabla c_{n-}\|_{L^2(\Omega)}^2 + \beta \|c_{n-}\|_{H^1(\Omega)}^2 + C(\beta) \|c_{n-}\|_{L^2(\Omega)}^2 \|\nabla\phi\|_{L^6(\Omega)}^4 \notag \\
&\leq \|c_{n-}\|_{L^2(\Omega)}^2 (1 + C \|\nabla\phi\|_{L^6(\Omega)}^4)\,, \notag
\end{align*}
where $\beta = \frac{1}{1+5M_0}$. Since $\|c_{n-}\|_{L^2(\Omega)}
= 0$ at $t=0$ and $1 + C \|\nabla\phi\|_{L^6(\Omega)}^4 \in
L^1((0,t_1))$, $c_{n-} \equiv 0$ i.e. $c_n \geq 0$.

Similarly, we may let
$c_{p-} = \min\{c_p,0\}$ and get
\begin{align*}
\frac{1}{2}\frac{d}{dt}\int_\Omega c_{p-}^2 dx &= -\int_\Omega \frac{1}{1+c_n^*} (|\nabla c_{p-}|^2 + c_{p-} \nabla \phi \cdot \nabla c_{p-}) dx \\
&\leq \|c_{p-}\|_{L^2(\Omega)}^2 (1 + C \|\nabla\phi\|_{L^6(\Omega)}^4), \notag
\end{align*}
which implies $c_{p-} \equiv 0$ i.e., $c_p \geq 0$. Now, we consider $u :=
c_n+c_p$ and $v := c_n-c_p$ which satisfy
\begin{align}
u_t &= \nabla \cdot (\nabla u - v\nabla \phi), \label{id59}
\end{align}
with boundary condition
\begin{align}
(\nabla u - v\nabla \phi)\cdot\nu &= 0. \label{id60}
\end{align}
To estimate the maximum of $u$, for $M \geq 2M_0$, we set
$u^{(M)} := \max\{u-M,0\}$ and $A_M(t):=\{x\in\Omega: u(x,t)>M\}$.
Multiply (\ref{id59}) by $u^{(M)}$ and take integration by parts. Then
\begin{align}
\frac{1}{2}\frac{d}{dt}\int_\Omega (u^{(M)})^2 dx &= -\int_\Omega |\nabla u^{(M)}|^2 dx + \int_{A_M(t)} (v \nabla \phi \cdot \nabla u^{(M)}) dx \notag\\
&\leq -\frac{1}{2}\int_\Omega |\nabla u^{(M)}|^2 dx + \frac{1}{2}\int_{A_M(t)} v^2 |\nabla \phi|^2 dx, \notag\\
&\leq -\frac{1}{2}\int_\Omega |\nabla u^{(M)}|^2 dx + \frac{1}{2}\int_{A_M(t)} u^2 |\nabla \phi|^2 dx. \label{id61}
\end{align}
For the last inequality in (\ref{id61}), we have used the fact that $u^2 \geq v^2$ because of $c_n,c_p \geq 0$. Hence
\begin{align}
\|u^{(M)}\|_{L^\infty((0,\tau);L^2(\Omega))}^2 + \|\nabla u^{(M)}\|_{L^2((0,\tau);L^2(\Omega))}^2 &\leq 2\int_0^{\tau} \int_{A_M(t)} u^2 |\nabla \phi|^2 dx dt \notag \\
&= 2\int_0^{\tau} \int_{A_M(t)} (u^{(M)} + M)^2 |\nabla \phi|^2 dx dt \notag\\
&\leq 4 \int_0^{\tau} \int_{A_M(t)} ((u^{(M)})^2 + M^2) |\nabla \phi|^2 dx dt \label{id62}
\end{align}
for $\tau \in (0,t_1)$. For simplicity, we employ some notations used
in \cite{LSU1968} that
\begin{align*}
Q_s &:= \Omega\times (0,s), \\
V_2(Q_s) &:= L^\infty((0,s);L^2(\Omega)) \cap L^2((0,s);H^1(\Omega)),
\end{align*}
and
\begin{align*}
\|w\|_{Q_s} := \|w\|_{L^\infty((0,s);L^2(\Omega))} + \|w\|_{L^2((0,s);H^1(\Omega))}
\end{align*}
for $w \in V_2(Q_s)$. In addition, we have the embedding
\begin{align}
\|w\|_{L^r((0,s);L^q(\Omega))} \leq C_s \|w\|_{Q_s} \label{id63}
\end{align}
for all $w \in V_2(Q_s)$, where $1/r + d/2q = d/4$, and
\begin{align*}
C_s = \beta_0 + (s^{d/2}|\Omega|^{-1})^{\frac{1}{2}-\frac{1}{q}}
\end{align*}
with $\beta_0$ depends only on $q,r,d$, and $\Omega$. Notice that the constant $C_s$ for (\ref{id63}) is increasing in $s$, then for $0 < s \leq t_1$, we can use the same constant $C_{t_1}$ such that
\begin{align}
\|w\|_{L^r((0,s);L^q(\Omega))} \leq C_{t_1} \|w\|_{Q_s}, \label{id63_1}
\end{align}
where $C_{t_1}$ is the constant in (\ref{id63}) with domain $Q_{t_1}$. Now, from
(\ref{id62}), we have
\begin{align}
\|u^{(M)}\|_{Q_{\tau}}^2 \leq C \int_0^{\tau} \int_{A_M(t)} ((u^{(M)})^2 + M^2) |\nabla \phi|^2 dx dt \label{id64}
\end{align}
for $0 < \tau < t_1$, where $C$ is a positive constant independent of $c_n,c_p,u,v,\phi,M$, and $\tau$. We will use $C$ to denote constants that may vary from line to line, but they are independent of $c_n,c_p,u,v,\phi,M$, and $\tau$. Then, by H\"{o}lder's inequality,
\begin{align}
\|u^{(M)}\|_{Q_{\tau}}^2 &\leq C \||\nabla \phi|^2\|_{L^3(Q_{\tau}(M))} \left(\|(u^{(M)})^2\|_{L^{\frac{3}{2}}(Q_{\tau}(M))} + \|M^2\|_{L^{\frac{3}{2}}(Q_{\tau}(M))}\right) \notag\\
&= C \|\nabla \phi\|_{L^6(Q_{\tau}(M))}^2 \left(\|u^{(M)}\|_{L^{3}(Q_{\tau}(M))}^2 + M^2 |Q_{\tau}(M)|^\frac{2}{3}\right) \label{id65}
\end{align}
where $Q_{\tau}(M) := Q_{\tau}\cap\{(x,t)\in Q_{\tau}: u(x,t)>M\}$. Here, by using the same inequality as (\ref{id19-2}) and (\ref{id26}), we have
\begin{align*}
\|\nabla \phi\|_{L^6(\Omega)}^2 &\leq C(\|v\|_{L^2(\Omega)}^2 + \|\phi_1\|_{H^2(\Omega)}^2) \\
&\leq C(\|u_0\|_{L^2(\Omega)}^2 + \|v_0\|_{L^2(\Omega)}^2)\cdot\exp\left\{K_2\int_0^{t_1} (1 + (\|v\|_{L^2(\Omega)}^2 + \|\phi_1\|_{H^2(\Omega)}^2)^2) ds\right\} \\
&\leq C(\|u_0\|_{L^2(\Omega)}^2 + \|v_0\|_{L^2(\Omega)}^2)\cdot\exp\left\{K_2\int_0^{t_1} (1 + 2\|v\|_{L^2(\Omega)}^4 + 2\|\phi_1\|_{H^2(\Omega)}^4) ds\right\} \\
&\leq C(\|u_0\|_{L^2(\Omega)}^2 + \|v_0\|_{L^2(\Omega)}^2)\cdot\exp\left\{K_2 \Big[t_1(1+2\|\phi_1\|_{H^2(\Omega)}^4) + 2R^4 \Big]\right\},
\end{align*}
where $R$ is the radius of the ball in $X$ where we obtain the solution $(c_n,c_p)$ as a fixed point of $F$. That is, $\nabla\phi\in L^\infty((0,t_1);L^6(\Omega))\subset L^6(Q_{t_1})$. Moreover,
\begin{align*}
\|\nabla \phi\|_{L^6(Q_{\tau}(M))}^2 &\leq \|\nabla \phi\|_{L^6(Q_{\tau})}^2 \notag\\
&\leq C\tau^{\frac{1}{3}}(\|u_0\|_{L^2(\Omega)}^2 + \|v_0\|_{L^2(\Omega)}^2) \\
&\quad\ \ \cdot\exp\left\{K_2 \Big[t_1(1+2\|\phi_1\|_{H^2(\Omega)}^4) + 2R^4 \Big]\right\}. \notag
\end{align*}
Hence, (\ref{id65}) implies
\begin{align}
\|u^{(M)}\|_{Q_{\tau}}^2 &\leq C\tau^{\frac{1}{3}}(\|u_0\|_{L^2(\Omega)}^2 + \|v_0\|_{L^2(\Omega)}^2)\cdot\exp\left\{K_2 \Big[t_1(1+2\|\phi_1\|_{H^2(\Omega)}^4) + 2R^4 \Big]\right\} \notag\\
&\quad\ \ \cdot\left(\|u^{(M)}\|_{L^{3}(Q_{\tau}(M))}^2 + M^2 |Q_{\tau}(M)|^\frac{2}{3}\right). \label{id65-1}
\end{align}
We now estimate the norm of $u^{(M)}$ in the right-hand side of
(\ref{id65-1}) by H\"{o}lder's inequality and (\ref{id63}),
\begin{align*}
\|u^{(M)}\|_{L^3 (Q_{\tau}(M))} &\leq \|u^{(M)}\|_{L^{2+\frac{4}{d}} (Q_{\tau}(M))} |Q_{\tau}(M)|^{\frac{4-d}{12+6d}} \notag \\
&\leq C_{t_1} \|u^{(M)}\|_{Q_{\tau}} |Q_{\tau}(M)|^{\frac{4-d}{12+6d}}. 
\end{align*}
Notice that
\begin{align*}
|Q_{\tau}(M)| \leq |Q_{\tau}| = \tau|\Omega|.
\end{align*}
Thus, for (\ref{id65-1}), if $\tau \leq \tau_0$, where
\begin{align*}
\tau_0 := \min\bigg\{t_1,&\bigg(\frac{1}{2}C^{-1}C_{t_1}^{-2}(\|u_0\|_{L^2(\Omega)}^2 + \|v_0\|_{L^2(\Omega)}^2)^{-1} |\Omega|^{-\frac{4-d}{6+3d}} \\
&\quad\ \ \cdot\exp\left\{-K_2 \Big[t_1(1+2\|\phi_1\|_{H^2(\Omega)}^4) + 2R^4 \Big]\right\}\bigg)^{\frac{2+d}{2}}\bigg\},
\end{align*}
i.e.,
\begin{align*}
CC_{t_1}^2\tau_0^{\frac{1}{3}}(\|u_0\|_{L^2(\Omega)}^2 + \|v_0\|_{L^2(\Omega)}^2)\cdot\exp\left\{K_2 \Big[t_1(1+2\|\phi_1\|_{H^2(\Omega)}^4) + 2R^4 \Big]\right\} |Q_{\tau_0}|^{\frac{4-d}{6+3d}} \leq \frac{1}{2},
\end{align*}
then we have
\begin{align*}
\|u^{(M)}\|_{Q_{\tau}}^2 &\leq 2CM^2\tau^{\frac{1}{3}}(\|u_0\|_{L^2(\Omega)}^2 + \|v_0\|_{L^2(\Omega)}^2) |Q_{\tau}(M)|^\frac{2}{3} \\
&\quad\ \ \cdot\exp\left\{K_2 \Big[t_1(1+2\|\phi_1\|_{H^2(\Omega)}^4) + 2R^4 \Big]\right\} \\
&\leq 2CM^2 t_1^{\frac{1}{3}}(\|u_0\|_{L^2(\Omega)}^2 + \|v_0\|_{L^2(\Omega)}^2) |Q_{\tau}(M)|^\frac{2}{3}\\
&\quad\ \ \cdot\exp\left\{K_2 \Big[t_1(1+2\|\phi_1\|_{H^2(\Omega)}^4) + 2R^4 \Big]\right\}.
\end{align*}
By Theorem 6.1 of Chapter II in \cite{LSU1968}, we have
\begin{align*}
\|u\|_{L^\infty(Q_{\tau})} \leq 4M_0 \left(1 + \tilde{C}\tau^{\frac{1}{3}}\right)
\end{align*}
for $0 < \tau \leq \tau_0$, where
\begin{align*}
\tilde{C} = 2^{\frac{2}{\kappa}+\frac{1}{\kappa^2}} |\Omega|^{\frac{1}{3}} C_{t_1}^{1+\frac{1}{\kappa}} &\bigg(2C t_1^{\frac{1}{3}}(\|u_0\|_{L^2(\Omega)}^2 + \|v_0\|_{L^2(\Omega)}^2) \\
&\quad\ \ \cdot\exp\left\{K_2 \Big[t_1(1+2\|\phi_1\|_{H^2(\Omega)}^4) + 2R^4 \Big]\right\}\bigg)^{\frac{1}{2}\left(1+\frac{1}{\kappa}\right)}
\end{align*}
with $\kappa = \frac{4-d}{3d}$. Therefore, $u(x,t) \leq 5M_0$ for $0<t<t_0$, where $t_0 = \min\{\tau_0,1/(4\tilde{C})^3\}$.
\end{proof}

By Lemma \ref{lem4}, $c_n^* = c_n$ and $c_p^* = c_p$ for $0<t<t_0$ in
(\ref{id55})--(\ref{id56}). Moreover, we have $c_n,c_p \in L^\infty(Q_{t_0})$.

Therefore, we may conclude the following theorem: \\

\begin{thm}\label{thm1}
Suppose that the initial data $c_{n,0}$ and $c_{p,0}$ satisfy (\ref{id7}) and (\ref{id8}), respectively. Then there exists $t_0 >0$ (depending on $\|c_{n,0}\|_{L^\infty(\Omega)}$, $\|c_{p,0}\|_{L^\infty(\Omega)}$,
$\|\phi_1\|_{H^2(\Omega)}$, $d$, $\alpha$, and $\Omega$) such that the system
(\ref{id1})--(\ref{id8}) has a solution $(c_n,c_p,\phi)$ with $0\le c_n,c_p \in
L^\infty((0,t_0);L^\infty(\Omega)) \cap L^2((0,t_0);H^1(\Omega))$ and $\frac{\partial c_n}{\partial t},\frac{\partial c_p}{\partial t} \in L^2((0,t_0);H^{-1}(\Omega))$.\\
\end{thm}

\begin{rmk}
By using the Moser iteration method, we have another approach to estimate the upper bound of $c_n + c_p$. We can rewrite (\ref{id59}) to be
\begin{align}
u_t &= \nabla \cdot (\nabla u - u V), \label{m1}
\end{align}
where
\begin{align*}
V = \left\{ \begin{array}{ll}
\displaystyle\frac{v}{u} \nabla\phi & \mbox{if } u \neq 0,\\
0 & \mbox{if } u = 0.
\end{array}\right.
\end{align*}
Note that we have proved that $c_n, c_p \geq 0$ in Lemma \ref{lem4}, then $|v/u| \leq 1$ for $u \neq 0$.
Moreover, by (\ref{id19-2}) and Lemma \ref{lem1}, $V \in L^\infty((0,t_1);L^6(\Omega))$. Set $w = u^\theta$ for $\theta > 1$. From (\ref{m1}), we deduce that
\begin{align}
\frac{1}{2}\frac{d}{dt}\int_\Omega w^2 dx &= -\int_\Omega \left(|\nabla w| ^2 + \theta(\theta-1)u^{\theta-2}w|\nabla u^2| - (2\theta-1) wV\cdot\nabla w \right)dx \notag\\
&\leq -\frac{1}{2} \int_\Omega |\nabla w| ^2 dx + \frac{1}{2} (2\theta-1)^2 \int_\Omega w^2|V|^2 dx \notag\\
&\leq -\frac{1}{2} \int_\Omega |\nabla w| ^2 dx + \frac{1}{2} (2\theta-1)^2 \|V\|_{L^6(\Omega)}^2 \|w\|_{L^3(\Omega)}^2. \label{m4}
\end{align}
Then for $0 < \tau < t_1$, we have
\begin{align*}
\int_\Omega w^2 dx \bigg|_{t=\tau} + \int_{Q_{\tau}} |\nabla w|^2 dxdt &\leq (2\theta-1)^2 \|V\|_{L^6(Q_{\tau})}^2 \|w\|_{L^3(Q_{\tau})}^2 + \int_\Omega w^2 dx \bigg|_{t=t_0} \notag\\
&\leq \mu^2 \theta^2  \|w\|_{L^3(Q_{\tau})}^2 + (2M_0)^{2\theta} |\Omega|,
\end{align*}
where $\mu = 2\|V\|_{L^6(Q_{\tau})}^2$. If $\lambda > 0$ satisfies $\frac{1}{3(1+\lambda)}(1+\frac{d}{2}) = \frac{d}{4}$, then by (\ref{id63_1})
\begin{align*}
\|w\|_{L^{3(1+\lambda)}(Q_{\tau})} \leq C_{t_1} \|w\|_{Q_{\tau}},
\end{align*}
where $C_{t_1}$ is the constant in (\ref{id63_1}) with $q = r = 3(1+\lambda)$. Thus, (\ref{m4}) implies
\begin{align}
\|w\|_{L^{3(1+\lambda)}(Q_{\tau})} &\leq 2C_{t_1} \left\{ \mu \theta \|w\|_{L^3(Q_{\tau})} + (2M_0)^\theta |\Omega|^{1/2} \right\} \notag\\
&= 2C_{t_1} \left\{ \mu \theta \|w^{\frac{1}{1+\lambda}}\|_{L^{3(1+\lambda)}(Q_{\tau})}^{1+\lambda} + (2M_0)^\theta |\Omega|^{1/2} \right\}. \label{m6}
\end{align}
Set
\begin{align*}
\Phi_k = \|u^{(1+\lambda)^k}\|_{L^{3(1+\lambda)}(Q_{\tau})} = \|u\|_{L^{3(1+\lambda)^k}(Q_{\tau})}^{(1+\lambda)^k},
\end{align*}
then by letting $\theta = (1+\lambda)^k$, (\ref{m6}) becomes
\begin{align}
\Phi_k \leq 2C_{t_1} \left\{ \mu (1+\lambda)^k \Phi_{k-1}^{1+\lambda} + (2M_0)^{(1+\lambda)^k} |\Omega|^{1/2} \right\}. \label{m7}
\end{align}
From the recursion inequalities (\ref{m7}), one can use induction to deduce that
\begin{align*}
\Phi_k &\leq (4C_{t_1} (1+\lambda))^{\frac{(1+\lambda)^k-1}{\lambda}} (1+\lambda)^{\frac{(1+\lambda)^k-1}{\lambda^2} - \frac{k}{\lambda}} \\
&\quad\quad\cdot\max\left\{ \mu^{\frac{(1+\lambda)^k-1}{\lambda}} \Phi_0^{(1+\lambda)^k},\max\Big\{\mu^{\frac{(1+\lambda)^k-1}{\lambda}},1\Big\} (2\max\big\{|\Omega|^{\frac{1}{2}},1\big\}M_0)^{(1+\lambda)^k} \right\}. \notag
\end{align*}
Therefore,
\begin{align*}
\|u\|_{L^\infty(Q_{\tau})} &= \lim_{k\rightarrow\infty} \Phi_k^{(1+\lambda)^{-k}} \notag\\
&\leq (4C_{t_1})^{\frac{1}{\lambda}} (1+\lambda)^{\frac{1}{\lambda}+\frac{1}{\lambda^2}} \\
&\quad\quad\cdot\max\left\{ \mu^{\frac{1}{\lambda}} \Phi_0, 2\max\Big\{\mu^{\frac{1}{\lambda}},1\Big\}\cdot\max\big\{|\Omega|^{\frac{1}{2}},1\big\}M_0 \right\}. \notag
\end{align*}
This provides an estimate of upper bound of $u = c_n + c_p$.
\end{rmk}

In the next section, we do compare the modified PNP to the classical PNP in numerical results.
\section{Numerical Experiments}\label{sec4}
In this section, we discuss on numerical results of modified PNP (\ref{eqn:MPNPsystem_2.1})--(\ref{eqn:MPNPsystem_2.3})  comparing with those of PNP (\ref{eqn:PNPsystem_0.1})--(\ref{eqn:PNPsystem_0.3}).
The computational domain is $[ -1, 1 ]$ for numerical experiments.
Mesh size is fixed with $h = 2^{-7}$ and time step size with $dt=10^{-3}$ throughout numerical experiments.

In time discretization, the backward Euler is used as follows:
\begin{eqnarray}
 \frac{\Cnkp-\Cnk}{dt} & = \nabla \cdot & \left \{ \frac{D (1+\Cnkp)}{1+\Cnkp+\Cpkp} \left( \nabla\Cnkp+\frac{z_n q}{k_B T}\Cnkp \nabla\phip \right) \right.\nonumber\\
 & &\left.+\frac{D\Cnkp}{1+\Cnkp+\Cpkp} \left( \nabla\Cpkp+\frac{z_p q}{k_B T}\Cpkp \nabla\phip \right) \right \}\label{eqn:NE1}\\
 \frac{\Cpkp-\Cpk}{dt} & = \nabla \cdot & \left \{ \frac{D (1+\Cpkp)}{1+\Cnkp+\Cpkp} \left( \nabla\Cpkp+\frac{z_p q}{k_B T}\Cpkp \nabla\phip \right) \right.\nonumber\\
 & &\left.+\frac{D\Cpkp}{1+\Cnkp+\Cpkp} \left( \nabla\Cnkp+\frac{z_n q}{k_B T}\Cnkp \nabla\phip \right) \right \}
 \label{eqn:NE11}\\
 \nabla \cdot \left( \varepsilon \nabla \phip \right) & = & -z_n q \Cnkp -z_p q \Cpkp, \label{eqn:NE12}
\end{eqnarray}
for $k = 0, 1, \cdots$ with initial data $c_n^0$ and $c_p^0$.
We set no-flux boundary conditions for charge densities and Dirichlet boundary condition for the electrostatic potential,
\begin{eqnarray}\label{eqn:NE2}
 \phip(-1) = \phi^0(-1), \quad \phip(1) = \phi^0(1) \quad \mbox{for } k=0,1,2,\cdots.
\end{eqnarray}
The edge averaged finite element (EAFE) method and the finite element method with piecewise linear basis functions
are used to solve Nernst-Planck equations (\ref{eqn:NE1}), (\ref{eqn:NE11}) and Poisson equation (\ref{eqn:NE12}), respectively~\cite{XuZi99}.
The variational formulation of the modified PNP (\ref{eqn:MPNPsystem_2.1})--(\ref{eqn:MPNPsystem_2.3}) is given by
\begin{eqnarray}\label{eqn:NE3}
 \left( \Cnkp,\xi \right) & + & dt \left( \frac{D(1+\Cnkp)}{1+\Cnkp+\Cpkp}\left(\nabla\Cnkp+\frac{z_n q}{k_B T}\Cnkp \nabla\phip\right), \nabla \xi \right) \nonumber\\
 & + & dt \left( \frac{D\Cnkp}{1+\Cnkp+\Cpkp}\left( \nabla\Cpkp+\frac{z_p q}{k_B T}\Cpkp \nabla\phip \right), \nabla \xi \right) \\
 & = & \left( \Cnk,\xi \right), \nonumber\\
 \left( \Cpkp,\eta \right) & + & dt \left( \frac{D(1+\Cpkp)}{1+\Cnkp+\Cpkp}\left(\nabla\Cpkp+\frac{z_p q}{k_B T}\Cpkp \nabla\phip\right), \nabla \eta \right) \nonumber\\
 & + & dt \left( \frac{D\Cpkp}{1+\Cnkp+\Cpkp}\left( \nabla\Cnkp+\frac{z_n q}{k_B T}\Cnkp \nabla\phip \right), \nabla \eta \right) \\
 & = & \left( \Cpk,\eta \right), \nonumber\\
 \varepsilon \left( \Delta \phip, \zeta \right) & = & -\left( z_n q \Cnkp +z_p q \Cpkp, \zeta \right).
\end{eqnarray}
However, the numerical computation of (\ref{eqn:NE1})--(\ref{eqn:NE12}) is not an easy task,
especially solving the charge concentration and the electrostatic potential both at the same time.
To over come the drawback, we apply a sub-updating iterative step because Poisson equation is not in time scale,
that is, the electrostatic potential should be simultaneously updated with the charge density in time.

Let $D_{n,n}=\frac{D(1+\Cnkpm)}{1+\Cnkpm+\Cpkpm}$,  $D_{n,p}=\frac{D\Cnkpm}{1+\Cnkpm+\Cpkpm}$, $D_{p,p}=\frac{D(1+\Cpkpm)}{1+\Cnkpm+\Cpkpm}$, and $D_{p,n}=\frac{D\Cpkpm}{1+\Cnkpm+\Cpkpm}$. Then the sub-updating numerical scheme with the index $m$ is given by
\begin{eqnarray}
 (\Cnkpmp,\xi) & + & dt \left( D_{n,n} \left(\nabla\Cnkpmp +\frac{z_n q}{k_B T}\Cnkpmp \nabla\phipm\right), \nabla \xi \right) \nonumber\\
 = (\Cnk,\xi)& - & dt \left( D_{n,p} \left( \nabla\Cpkpm +\frac{z_p q}{k_B T}\Cpkpm \nabla\phipm \right), \nabla \xi \right), \label{eqn:NE4}\\
 (\Cpkpmp,\eta) & + & dt \left( D_{p,p} \left(\nabla\Cpkpmp +\frac{z_p q}{k_B T}\Cpkpmp \nabla\phipm\right), \nabla \eta \right) \nonumber\\
 = (\Cpk,\eta)& - & dt \left( D_{p,n} \left( \nabla\Cnkpm +\frac{z_n q}{k_B T}\Cnkpm \nabla\phipm \right), \nabla \eta \right), \label{eqn:NE41}\\
 \varepsilon \left( \Delta \phipmp, \zeta \right) & = & -\left( z_n q \Cnkpmp + z_p q \Cpkpmp, \zeta \right)
 \label{eqn:NE42}
\end{eqnarray}
for $m = 0, 1, 2 \cdots$ letting $c_n^{k+1,0}=\Cnk$, $c_p^{k+1,0}=\Cpk$.
The boundary condition of the electrostatic potential is
\begin{eqnarray}\label{eqn:NE5}
 \phipmp(-1) = \phi^0(-1), \quad \phipmp(1) = \phi^0(1)\quad \mbox{for } k,m=0,1,2,\cdots.
\end{eqnarray}

\begin{rmk}
Developing numerical scheme satisfying energy law is another work in numerical computations.
The numerical discretization scheme for (\ref{eqn:NE4})--(\ref{eqn:NE42})
has a certain limitation for preserving energy law in finite dimensional space.
However, the comparison of dissipations $\triangle^\ast$, $\triangle$ of the modified and original PNP systems may provide the difference between two systems.
\\
\end{rmk}

In Figure~\ref{fig:F_1}, we present numerical results of initial data (top row), equilibrium states ${{c}_{n}},{{c}_{p}}$ (middle row) and $\phi $ (bottom row) for the modified and original PNP systems with boundary conditions $\phi(-1)=0.05$, $\phi(1)=0.0$ (left panel) and $\phi(-1)=0.0$, $\phi(1)=0.05$ (right panel), respectively. These results show that the modified and original PNP systems have the same equilibrium states even though they are totally different systems of equations. However, different dynamics of the modified and original PNP systems can be expressed by numerical results of $\triangle^{\ast}$ and $\triangle$ in time (see Figure~\ref{fig:F_2}) due to the extra term $\frac{k_BT}{D_{n,p}}c_nc_p|\vec{u}^*_n - \vec{u}^*_p|^2$ in the dissipation functional of the modified PNP system.

\begin{figure}[hbt]
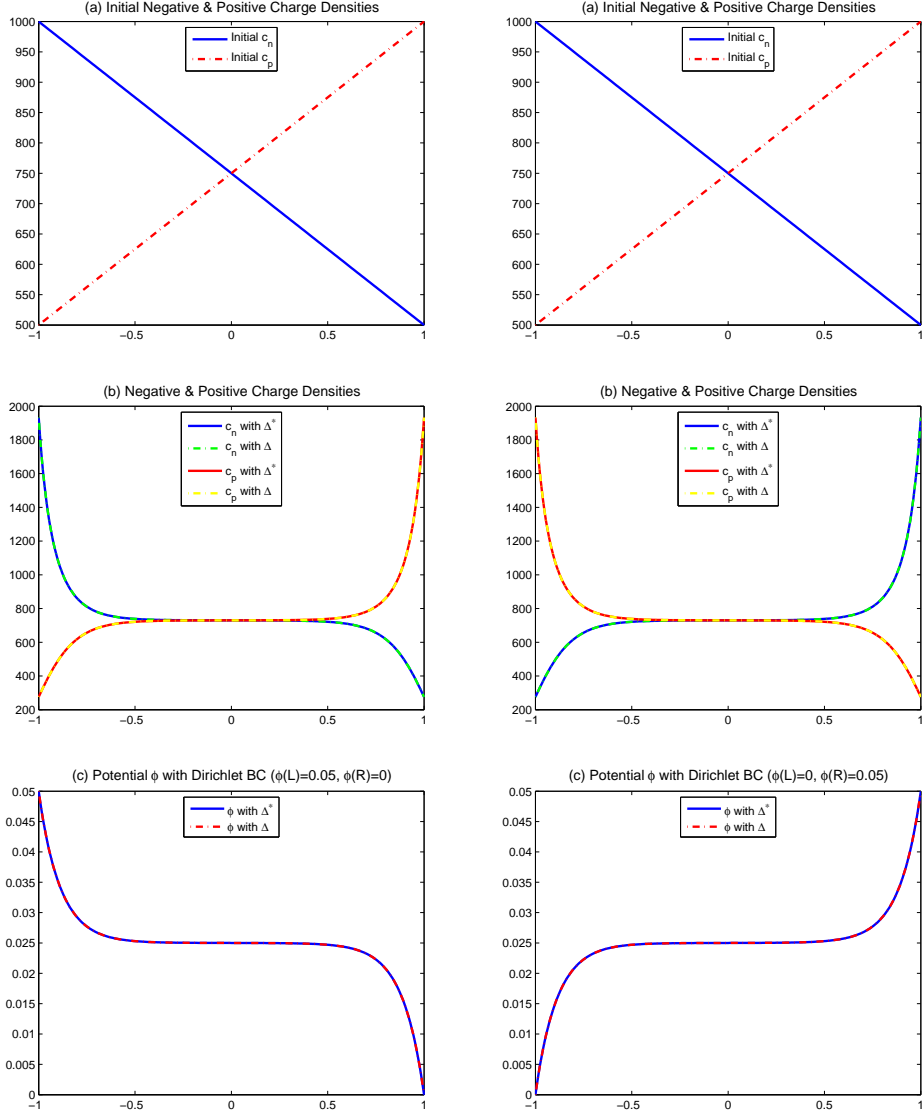

 \centering{%
  \begin{tabular}{@{\hspace{-0pc}}c@{\hspace{-0pc}}c}
   \psfig{figure=1_L_a.eps,width=2.6in} & \psfig{figure=1_R_a.eps,width=2.6in}\\
   \psfig{figure=1_L_b.eps,width=2.6in} & \psfig{figure=1_R_b.eps,width=2.6in}\\
   \psfig{figure=1_L_c.eps,width=2.6in} & \psfig{figure=1_R_c.eps,width=2.6in}
  \end{tabular}}
 \caption{\small The comparison of numerical results $c_n$, $c_p$, $\phi$ of the modified PNP system
 to those of original PNP system. Initial data (top row), charge densities (middle row), and the electrostatic potential (bottom row). The left panel is for the numerical results with the electrostatic potential boundary condition, $\phi(-1)=0.05$, $\phi(1)=0.0$, and the right one for the numerical results with $\phi(-1)=0.0$, $\phi(1)=0.05$.}\label{fig:F_1}
\end{figure}

\begin{figure}[hbt]
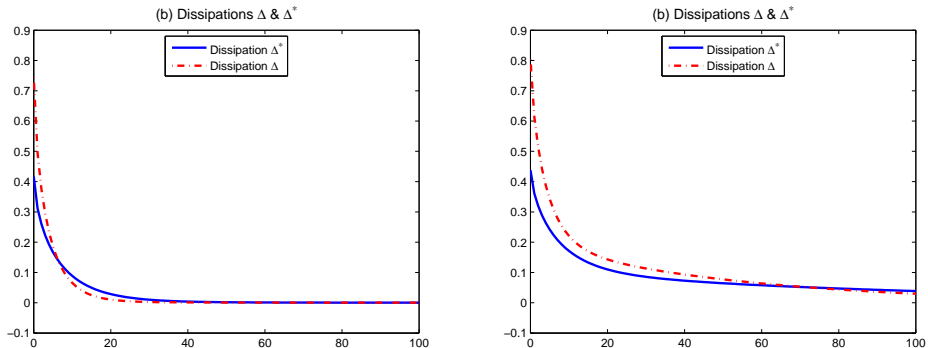

 \centering{%
  \begin{tabular}{@{\hspace{-0pc}}c@{\hspace{-0pc}}c}
   \psfig{figure=2_L_b.eps,width=2.6in} & \psfig{figure=2_R_b_100.eps,width=2.6in}
  \end{tabular}}
 \caption{\small The comparison of the dissipation
 of the modified PNP system to those of the original PNP system.
 The left panel is for the numerical results with the electrostatic potential boundary condition, $\phi(-1)=0.05$, $\phi(1)=0.0$, and the right one for the numerical results with $\phi(-1)=0.0$, $\phi(1)=0.05$.}\label{fig:F_2}
\end{figure}

\section{Conclusion}\
By employing an energetic variational approach, we derive a modified PNP system to describe the dynamics of non-ideal ions, such as those with relatively high concentrations.
In this work, we maintain the  energy functional as the original PNP system but modify dissipation functional with an additional dissipation term,
which accounts for  the relative velocity fields of different ion species. The modified PNP system is highly coupled and may even involve degenerate parabolicity in the system.
The analysis and simulation of such a system become much more involved than the original PNP system. As one preliminary step, we develop (with rigorous proof) the local existence theorem of this modified PNP system. By comparing the numerical results of the modified PNP system and the original PNP system, we verify that these two systems have the same equilibrium states but with different dynamics because of different dissipations. In the following up work, we are including modifications to
both free energy functional and the dissipation functional, and study the resulting PNP-type system theoretically and numerically.

\section{Acknowledgment}\
Chia-Yu Hsieh wishes to express sincere thanks to the Department of Mathematics of Pennsylvania State University for the chance of one-year visit.
YunKyong Hyon is partially supported by the National Institute for Mathematical Sciences (NIMS) grant
funded by the Korea government (No. B21401).
Tai-Chia Lin is partially supported by the National Science Council of Taiwan grants NSC-102-2115-M-002-015 and NSC-100-2115-M-002-007.
Chun Liu is partially support by the NSF grants DMS-1109107, DMS-1216938, and DMS-1159937.




\begin{thebibliography}{10}

\bibitem{AbMa78} {\sc R. Abraham and J.E. Marsden}, {\em Foundations of Mechanics, Second Edition}, Addison-Wesley, 1978.

\bibitem{Ar89} {\sc V.I. Arnold}, {\em Mathematical Methods of Classical Mechanics, Second
Edition}, Springer-Verlag, New York, 1989.

\bibitem{B1992} {\sc P. Biler}, {\em Existence and Asymptotics of Solutions for a Parabolic-Elliptic System with Nonlinear No-Flux Boundary Conditions}, Nonlinear Analysis, 19(21):1121--1136, 1992.

\bibitem{BHN1994} {\sc P. Biler and W. Hebisch and T. Nadzieja}, {\em The Debye System: Existence and Large Time Behavior of Solutions}, Nonlinear Analysis, 23(9):1189--1209, 1994.

\bibitem{DSZ_pp2008} {\sc M.A. Dorf, V.E. Semenov, and V.G. Zorin}, {\em A fluid model for ion heating due to ionization in a plasma flow}, Phys. of Plasmas, 15, 093501 (1-6), 2008.

\bibitem{EHL10} {\sc B. Eisenberg, Y. Hyon, and Chun Liu}, {\em Energy Variational Analysis of Ions in Water and Channels: Field Theory for Primitive Models of Complex Ionic Fluids}, J. Chem. Phys., 133(10), 104104, 2010.

\bibitem{Er1} B. Eisenberg, {\em Mass action in ionic
solutions}, Chemical Physics Letters, 511, 1-6, 2011.

\bibitem{Er11} B. Eisenberg, {\em Crowded Charges in Ion Channels, Advances in Chemical Physics}, John Wiley and Sons, Inc., 77-223, 2011.

\bibitem{E1} B. Eisenberg, {\em A Leading Role for Mathematics in the Study of Ionic Solutions}, SIAM News, 45, 11-12,
\bibitem{G-zm1985} H. Gajewski, {\em On Existence, Uniqueness and Asymptotic Behavior of Solutions of the
Basic Equations for Carrier Transport in Semiconductors}, Z. Angow. Math. Mech. 66 (1985) 2, 101-108.

\bibitem{GT1983} {\sc D. Gilbarg and N.S. Trudinger},
{\em Elliptic Partial Differential Equations of Second Order, Second
Edition}, Springer-Verlag, Berlin, 1983.


\bibitem{HLLE} T.L. Horng, T.C. Lin, C. Liu and B. Eisenberg,
{\em PNP equations with steric effects: a Model of Ion Flow through Channels}, J. Phys Chem B, 116(37), 11422-11441, 2012.

\bibitem{HEL} Y. Hyon, D. Y. Kwak and C. Liu,
{\em A Mathematical model for the hard sphere repulsion in ionic solutions},
Commun. Math. Sci., 9(2), 459--475, 2011.

\bibitem{HDL} Y. Hyon, B. Eisenberg and C. Liu,
{\em Energetic Variational Approach in Complex Fluids : Maximum Dissipation Principle},
DCDS-A, Vol. 26, No. 4, pp.1291--1304, 2010.

\bibitem{Ku76} {\sc R. Kubo}, {\em Thermodynamics: An Advanced Course with Problems and Solutions}, North-Holland Pub. Co., 1976.

\bibitem{LSU1968} {\sc O.A. Lady\v{z}enskaja and V.A. Solonnikov and N.N. Ural'ceva}, {\em Linear and Quasi-linear Equations of Parabolic Type}, Amer. Math. Society, Providence, 1968.

\bibitem{LHLL} {\sc C.C. Lee, H. Lee, Y. Hyon, T. C. Lin and C. Liu},
{\em New Poisson-Boltzmann Type Equations: One-Dimensional Solutions}, Nonlinearity 24 (2011) 431--458.

\bibitem{LE-cms13} T.C. Lin and B. Eisenberg, {\em A new approach to the Lennard-Jones potential and a new model: PNP-steric equations}, Comm. Math. Sci., Vol. 12, No. 1 (2014) 149-173.

\bibitem{On31} {\sc L. Onsager}, {\em Reciprocal Relations in Irreversible Processes. I.}, Phys. Rev., II. Ser., 37, 405--426, 1931.

\bibitem{On31a} {\sc L. Onsager}, {\em Reciprocal Relations in Irreversible Processes. II.}, Phys. Rev., II. Ser., 38, 2265--2279, 1931.

\bibitem{R1} {\sc R. Ryham}, {\em An Energetic Variational Approach To Mathematical Modeling Of Charged Fluids: Charge Phases, Simulation And Well Posedness}, thesis, Pennsylvania State University, 2006.

\bibitem{RLW} {\sc R. Ryham, C. Liu, and Z.Q. Wang}, {\em On Electro-Kinetic Fluids: One Dimensional Configurations}, Discrete Contin. Dyn. Syst. Ser. B 6 (2006), no.~2, p.~357--371.

\bibitem{RLZ} {\sc R. Ryham, C. Liu, and L. Zikatanov}, {\em An Mathematical Models for the Deformation of Electrolyte Droplets},
Discrete Contin. Dyn. Syst. Ser. B 8 (2007), no. 3, p. 649--661.

\bibitem{S-na1985} T. I. Seidman, {\em TIME-DEPENDENT SOLUTIONS OF A NONLINEAR SYSTEM ARISING IN SEMICONDUCTOR THEORY}, Nonlinear Analysis, Theory, Merhods Applications. Vol. 9, No. 11, pp.~1137-1157, 1985.

\bibitem{Ra73} {\sc J.W. Strutt}, {\em Some General Theorems Relating to Vibrations}, Proc. of L.M.S. IV, 357--368, 1873.

\bibitem{TK93} {\sc R. Taylor and R. Krishna}, {\em Multicomponent Mass Transfer}, Wiley, 1993.

\bibitem{XSL14} {\sc S. Xu, P. Sheng, and C. Liu}, {\em An Energetic Variational Approach for Ion Transport},
Commun. Math. Sci., Vol. 12, No. 4, Pages 779--789, (2014).

\bibitem{XuZi99} {\sc J. Xu, L. Zikatanov}, {\em A Monotone Finite Element Scheme for Convection-diffusion Equations},
Mathematics of Computation, Vol. 68, No. 228, Pages 1429--1446, (1999).



\end{thebibliography}
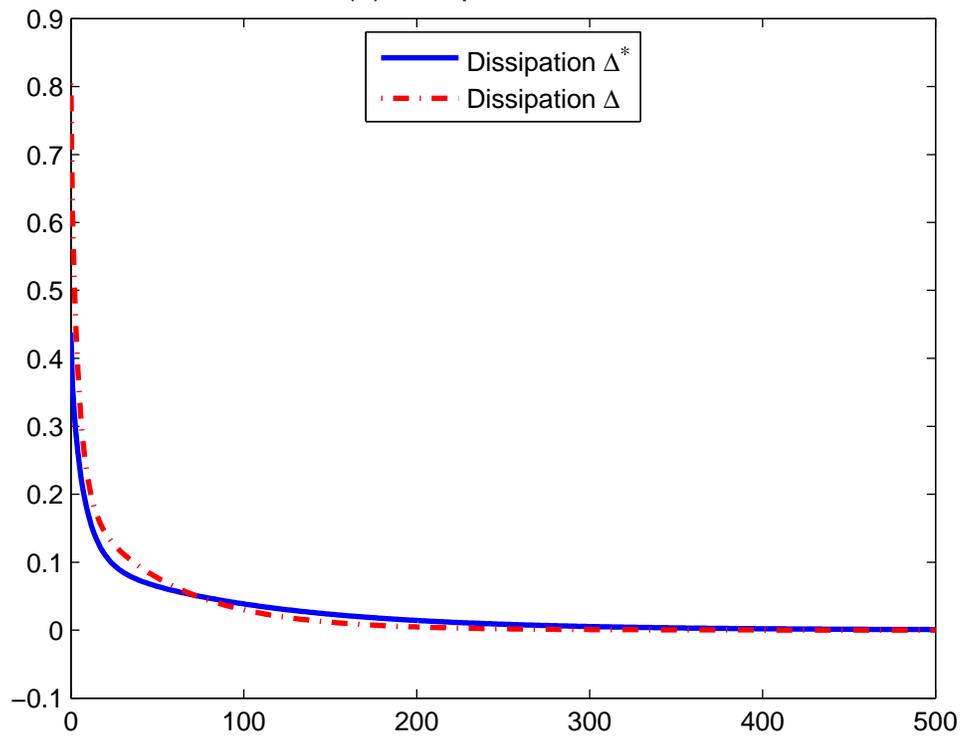
\end{document}